\documentclass[leqno, 11pt]{article}
\usepackage{times}
\usepackage{amssymb,amscd}
\usepackage{amsmath}
\usepackage{amsfonts}
\usepackage{mathrsfs}
\usepackage{xfrac}
\usepackage{comment}
\usepackage{enumerate}

\usepackage{amsthm}

\newtheorem{theorem}{Theorem}[section]
\newtheorem{definition}[theorem]{Definition}

\newtheorem{example}[theorem]{Example}
\newtheorem{lemma}[theorem]{Lemma}
\newtheorem{corollary}[theorem]{Corollary}

\begin{document}

\title{Some remarks on semantics and expressiveness of the Sentential Calculus with Identity}

\author{Steffen Lewitzka\thanks{Universidade Federal da Bahia -- UFBA,
Departamento de Ci\^encia da Computa\c c\~ao,
Instituto de Computação,
40170-110 Salvador BA,
Brazil,
steffenlewitzka@web.de}}

\date{}

\maketitle

\begin{abstract}
Suszko's Sentential Calculus with Identity $\mathit{SCI}$ results from classical propositional calculus $\mathit{CPC}$ by adding a new connective $\equiv$ and axioms for identity $\varphi\equiv\psi$ (which we interpret here as `propositional identity'). We reformulate the original semantics of $\mathit{SCI}$ in terms of Boolean prealgebras establishing a connection to `hyperintensional semantics'. Furthermore, we define a general framework of dualities between certain $\mathit{SCI}$-theories and Lewis-style modal systems in the vicinity of $\mathit{S3}$. Suszko's original approach to two $\mathit{SCI}$-theories corresponding to $\mathit{S4}$ and $\mathit{S5}$ can be formulated as a special case. All these dualities rely particularly on the fact that Lewis' `strict equivalence' is axiomatized by the $\mathit{SCI}$-principles of `propositional identity'.
\end{abstract}


\section{Introduction}

The set $Fm_\equiv$ of formulas of the Sentential Calculus with Identity $\mathit{SCI}$ is inductively defined in the usual way over an infinite set $V$ of propositional variables $x_0, x_1, ...$, logical connectives $\bot$, $\top$, $\neg$, $\vee$, $\wedge$, $\rightarrow$ and an identity connective $\equiv$ for building formulas of the form $(\varphi\equiv\psi)$.
As a deductive system, $\mathit{SCI}$ extends classical propositional logic $\mathit{CPC}$ by the identity axioms (id1)--(id7) below. That is, $\mathit{SCI}$ can be axiomatized by all formulas having the form of a classical tautology together with the following identity axioms:\\

\noindent (id1) $\varphi\equiv\varphi$\\
(id2) $(\varphi\equiv\psi)\rightarrow (\varphi\rightarrow\psi)$  \\
(id3) $(\varphi\equiv\psi)\rightarrow (\neg\varphi\equiv\neg\psi)$\\
(id4)--(id7) $((\varphi_1\equiv\psi_1)\wedge (\varphi_2\equiv\psi_2))\rightarrow ((\varphi_1*\varphi_2)\equiv (\psi_1 *\psi_2))$,\\ 
where $*\in\{\vee, \wedge, \rightarrow, \equiv\}$, respectively.\\

With Modus Ponens MP as inference rule, the notion of derivation is defined in the usual way. We write $\varPhi\vdash_{\mathit{SCI}}\varphi$ if there is a derivation of $\varphi\in Fm_\equiv$ from the set $\varPhi\subseteq Fm_\equiv$. The introduction of $\mathit{SCI}$ is a consequence of R. Suszko's work on non-Fregean logics which, in turn, was motivated by his attempts to formalize ontological aspects of Wittgenstein's \textit{Tractatus logico-philosophicus} (see, e.g. \cite{sus1}). Recall that, according to G. Frege, the denotation (referent, \textit{Bedeutung}) of a formula is nothing but a truth-value. This principle, called by Suszko the Fregean Axiom, can be formalized as $(\varphi\leftrightarrow\psi)\rightarrow (\varphi\equiv\psi)$ if we assume the classical interpretation of connectives and read $\varphi\equiv\psi$ as `$\varphi$ and $\psi$ have the same denotation'. The essential feature of a non-Fregean logic is the failure of Fregean Axiom. $\mathit{SCI}$ can be seen as a basic non-Fregean logic extending $\mathit{CPC}$. The identity axioms express our basic intuition on propositional identity: it should be a congruence relation on formulas that refines equivalence $\leftrightarrow$.\footnote{Indeed, $(\varphi\equiv\psi)\rightarrow (\varphi\leftrightarrow\psi)$ as well as $(\varphi\equiv\psi)\rightarrow (\psi\equiv\varphi)$ and $((\varphi\equiv\psi)\wedge (\psi\equiv\chi))\rightarrow (\varphi\equiv\chi)$ are derivable. The `compatibility' with connectives of the language is expressed by axioms (id3) and (id4)--(id7).} As already pointed out in \cite{blosus}, replacing (id3)--(id7) by the single scheme
\begin{equation}\label{10}
(\varphi\equiv\psi)\rightarrow (\chi[x:=\varphi]\equiv\chi[x:=\psi]),\footnote{$\chi[x:=\varphi]$ is the result of substituting $\varphi$ for every occurrence of variable $x$ in $\chi$. The concept can be formally defined in the obvious way by induction on the construction of $\chi$.} 
\end{equation} 
which we call the Substitution Principle SP, results in a deductively equivalent system.\footnote{This fact can be shown by induction on $\chi$.} SP essentially says that formulas with the same denotation can be replaced by each other in any context. This principle can be seen as a particular instance of a general ontological law known in the literature as the \textit{indiscernibility of identicals} or \textit{Leibniz's law}. In a formal context, SP represents a necessary condition for the existence of a natural propositional semantics. In fact, if we interpret logical connectives and further operators of the object language semantically as functions on propositions, then SP says that all these functions are well-defined: identical arguments yield identical function values. For instance, SP holds in classical and intuitionistic propositional logic with propositional identity $\varphi\equiv\psi$ given as equivalence $\varphi\leftrightarrow\psi$. If we assume the propositional modal language and define propositional identity as strict equivalence: $(\varphi\equiv\psi) := \square(\varphi\leftrightarrow\psi)$, then SP is a derivable principle in Lewis modal systems $\mathit{S3}$--$\mathit{S5}$ but not in the weaker systems $\mathit{S1}$ and $\mathit{S2}$, cf. \cite{lewjlc1, lewsl}. However, it is enough to add SP to system $\mathit{S1}$ in order to get a logic with a natural algebraic semantics. This logic was introduced in \cite{lewjlc1} under the name $\mathit{S1}$+$\mathit{SP}$. In the present paper, we shall refer to it by the simpler label $\mathit{S1SP}$. We then get the hierarchy $\mathit{S1SP} \subsetneq \mathit{S3}\subsetneq \mathit{S4}\subsetneq \mathit{S5}$ of Lewis (-style) modal logics for which we can use the same framework of algebraic semantics based on Boolean algebras (we shall explore this kind of semantics in section 4).

The interpretability of $\mathit{SCI}$ in Lewis system $\mathit{S3}$ indicates a strong connection between $\mathit{SCI}$-theories and Lewis-style modal systems. Essential aspects of that connection were already revealed by Suszko, Bloom \cite{sus, blosus} showing that specific extensions of $\mathit{SCI}$ correspond, in some sense, to modal logics $\mathit{S4}$ and $\mathit{S5}$, respectively. Instead of interpreting $\mathit{SCI}$-theories in Lewis modal systems, Suszko's approach restores modal logic within $\mathit{SCI}$-extensions via the definition $\square\varphi := (\varphi\equiv\top)$. 

In the present paper, we study dualities between $\mathit{SCI}$-theories and Lewis-style modal systems (not restricted to $\mathit{S4}$ and $\mathit{S5}$) in a systematical way and establish precise criteria for the existence of such dualities. We consider here both object languages separately -- the language of $\mathit{SCI}$ versus the language of propositional modal logic -- and define appropriate translations between them. In contrast to the original model-theoretic approach (cf. \cite{blosussl, blosus}), we introduce $\mathit{SCI}$-models explicitly as Boolean prealgebras (or Boolean prelattices). In this way, we find a bridge to an approach known in the literature as `hyperintensional semantics' (see, e.g., \cite{foxlap, pol}) and present $\mathit{SCI}$ as a basic classical logic for (hyper-) intensional modeling and reasoning. 

\section{Intensionality as a measure for the discernibility of propositions}

Originally introduced by M. J. Cresswell \cite{cres}, the notion of `hyperintensionality' has been interpreted in different ways in the literature and there seems to be no formal standard definition. Usually, an operator (of a given logic) is regarded as extensional if its application to formulas with the same truth-value results again in formulas having the same truth-value, otherwise the operator may be seen as intensional.\footnote{We consider here only classical logics.} In the context of possible worlds semantics, an operator is often regarded as hyperintensional if its application to formulas having the same truth-values at all possible (accessible) worlds does not necessarily result in formulas with the same truth-value at the actual world. For instance, the modal operator of normal modal logics is intensional (but not hyperintensional). Therefore, modal logics are often regarded as intensional logics. Possible worlds semantics, however, is not an appropriate framework for dealing with hyperintensional operators. There are proposals in the literature conceiving hyperintensional semantics in terms of Boolean prelattices (see, e.g., \cite{foxlap, pol, lewsl}), and we will follow a similar approach. For this purpose, let us regard a proposition as the denotation of a formula at a given model. A proposition can be, e.g., a truth-value (in classical propositional logic), a set of possible worlds (in normal modal logics), an element of some algebraic structure, etc. Under these assumptions, we propose to explain intensionality as a measure of discernibility of propositions. The more propositions can be distinguished in models of the underlying classical logic the higher the degree of intensionality. In $\mathit{CPC}$, only two propositions, the True and the False, can be distinguished. Current modal logics provide much more (infinitely many) propositions: even if two formulas $\varphi$ and $\psi$ have the same truth-value at the actual world, they may have different truth-values at some accessible world and thus denote different propositions: the Fregean Axiom does not hold -- the denotation of a formula is more than a classical truth-value. Nevertheless, many propositions remain indiscernible: logically equivalent formulas such as $\neg\neg\varphi$ and $\varphi$ will always denote the same proposition in classical modal logics. The aim of hyperintensional semantics is to overcome such limitations of the possible worlds framework (motivations come, e.g., from the study of natural language semantics) and to provide a more fine-grained approach that allows to discern even more propositions. This goal can be perfectly achieved working with $\mathit{SCI}$ and appropriate axiomatic extensions. By a proposition we will mean more specifically the element of a given $\mathit{SCI}$-model. The degree of intensionality of a model is the largest number of propositions that can be distinguished. We shall see that all expressible intensions can be discerned in logic $\mathit{SCI}$. In fact, there is an $\mathit{SCI}$-model where any two different formulas denote different propositions, see Theorem \ref{300} below. We call such a model intensional since the denotation of a formula can be identified with its intension, i.e. its syntactical form. In this sense, $\mathit{SCI}$ is a logic of highest degree of intensionality and, of course, is able to model hyperintensional operators. Imposing appropriate axioms, we get specific $\mathit{SCI}$-theories where specific propositions become indiscernible. In particular, $\mathit{CPC}$ as well as some Lewis-style modal logics can be represented as specific $\mathit{SCI}$-theories. While models of $\mathit{CPC}$ are extensional, models of modal logics lie somewhere between the extremes of extensional and intensional model. In the following, we will present $\mathit{SCI}$ as an (hiper-) intensional logic.\footnote{In contrast to our view, Bloom and Suszko explicitly deny the intensional character of $\mathit{SCI}$. ``Some people, upon discovering that the identity connective was not truth-functional, have thought that $\mathit{SCI}$ is an intensional logic. We emphatically deny this. The essence of intensionality is that the rule "equals may be replaced by equals" fails. However, this rule \textit{does} hold in the SCI ... " (cf. p. 1 of \cite{blosus}). Actually, that rule is formalized by SP which is valid in $\mathit{SCI}$.} 

\section{Boolean prealgebras as models of $\mathit{SCI}$}

Recall that a preorder on a set $M$ is a binary relation on $M$ satisfying reflexivity and transitivity. If a preorder is also antisymmetic, then it is a partial order. We also expect the reader to be familiar with the concepts of Boolean algebra, filters and ultrafilters (on Boolean algebras) and quotient Boolean algebras. We apply a somewhat unusual notation for Boolean algebras (with operators) which has the advantage that for any new connective or symbol of the underlying object language a corresponding operator for the algebraic semantics can easily be identified. In particular, for the connectives $\vee, \wedge, \neg, \bot, \top, \rightarrow$ of our classical logic, we denote the corresponding operations of a given Boolean (pre-) algebra $\mathcal{B}$ by $f_\vee, f_\wedge, f_\neg, f_\bot, f_\top, f_\rightarrow$ (or, more precisely, by $f_\vee^\mathcal{B}$, etc., if we wish to emphasize the given context of (pre-) algebra $\mathcal{B}$). 

\begin{definition}\label{100}
A structure $\mathcal{B}=(B, f_\vee, f_\wedge, f_\neg, f_\bot, f_\top, f_\rightarrow, \preceq)$ of type $(2,2,1,0,0,2)$ with a preorder $\preceq$ on universe $B$ is a Boolean prealgebra if the relation $\approx$ defined by $a\approx b$ $:\Leftrightarrow$ ($a\preceq b$ and $b\preceq a$) is a congruence relation on $B$ such that the quotient $\mathcal{B}/{\approx}$ is a Boolean algebra, and for all $a,b\in B$ we have: $a\preceq b\Leftrightarrow f_\wedge(a,b)\approx a$. In this case, we call $\approx$ the associated congruence, and we call quotient $\mathcal{B}/{\approx}$ the associated Boolean algebra. If the given structure $\mathcal{B}$ itself is a Boolean algebra, then we denote the underlying lattice order by $\le$.\footnote{Even if $\mathcal{B}$ is a Boolean algebra, the lattice order $\le$ may differ from the given preorder $\preceq$.}
\end{definition}

Of course, every Boolean algebra together with its lattice order (regarded as a preorder) is trivially a Boolean prealgebra. Recall that every Boolean algebra is a Heyting algebra. Those Heyting algebras which are not Boolean algebras are non-trivial though natural examples of Boolean prealgebras. In order to see this, consider any designated ultrafilter $U$ of a given Heyting algebra (which exists by Zorn's Lemma) and the preorder $a\preceq b$ $:\Leftrightarrow f_\rightarrow(a,b)\in U$, where $f_\rightarrow(a,b)$ is the relative pseudo-complement of $a$ w.r.t. $b$. Then the resulting quotient algebra modulo $\approx$ is the two-element Boolean algebra.\footnote{Of course, there may exist further congruence relations on a given Heyting algebra that result in a Boolean quotient algebra.} Considering the intuitionistic tautology $(x\rightarrow y)\leftrightarrow (x\rightarrow (x\wedge y))$, one easily checks that also the condition $a\preceq b\Leftrightarrow f_\wedge(a,b)\approx a$ holds for all elements $a,b$ of the Heyting algebra.\footnote{Recall that all intuitionistic tautologies are interpreted by the top element $f_\top$ of any Heyting algebra under any assignment, and also recall that the following condition is valid in every Heyting algebra: $f_\rightarrow(a,b)=f_\top$ iff $a\le b$.}

Note that we cannot do without that second condition in Definition \ref{100}. Even if the resulting quotient $\mathcal{B}/{\approx}$ of structure $\mathcal{B}$ is a Boolean algebra, condition $a\preceq b\Leftrightarrow f_\wedge(a,b)\approx a$ is not necessarily true. Consider, for instance, the $4$-element Boolean algebra $Pow(2)$ with the preorder $\preceq$ given by set-theoretic inclusion on $Pow(2)$ extended by the tuple $(\{1\},\{2\})$, so we have in particular $\{1\}\preceq\{2\}$. Relation $\approx$ is the identity on $Pow(2)$ and the resulting quotient algebra is, of course, again the Boolean algebra $Pow(2)$. However, the second condition of Definition \ref{100} fails since we have $\{1\}\preceq\{2\}$, but $\{1\}\subsetneq\{2\}$, i.e. $\{1\}\cap\{2\}\neq\{1\}$.

If one deals with Boolean algebras, then one usually considers only the operations of supremum (join) $f_\vee$, infimum (meet) $f_\wedge$, complement $f_\neg$, least element $f_\bot$ and greatest element $f_\top$. Further relevant operations, such as implication $f_\rightarrow(a,b) :=f_\vee(f_\neg(a),b)$, are definable. This, however, does not hold in general for Boolean prealgebras. For instance, although we have $f_\rightarrow(a,b)\approx f_\vee(f_\neg(a),b)$, the propositions (i.e. elements) $f_\rightarrow(a,b)$ and $f_\vee(f_\neg(a),b)$ may be distinct. 

\begin{lemma}\label{122}
Let $\mathcal{B}$ be a Boolean prealgebra with preorder $\preceq$, and let $\mathcal{B}/{\approx}$ be the associated Boolean algebra with lattice order $\le^{\mathcal{B}/\approx}$. Then for all $a,b\in B$: $a\preceq b\Leftrightarrow\overline{a}\le^{\mathcal{B}/\approx}\overline{b}$.
\end{lemma}

\begin{proof}
Let $\mathcal{B}$ be a Boolean prealgebra with preorder $\preceq$. Then for all $a,b\in B$, $a\preceq b \Leftrightarrow f^\mathcal{B}_\wedge(a,b)\approx a \Leftrightarrow f_\wedge^{\mathcal{B}/\approx}(\overline{a},\overline{b})=\overline{a}\Leftrightarrow\overline{a}\le^{\mathcal{B}/\approx} \overline{b}$.
\end{proof}

\begin{definition}\label{126}
Let $\mathcal{B}$ be a Boolean prealgebra with associated Boolean algebra $\mathcal{B}/{\approx}$, and let $F\subseteq B$ be closed under $\approx$, i.e. $a\in F\Leftrightarrow b\in F$ whenever $a\approx b$, for any $a,b\in B$. Then we say that $F$ is a filter of $\mathit{B}$ if the set $F^{\mathcal{B}/{\approx}}=\{\overline{a}\mid a\in F\}$ is a filter (in the usual sense) of Boolean algebra $\mathcal{B}/{\approx}$. The notions of proper filter and ultrafilter of a Boolean prealgebra are defined analogously.
\end{definition}

\begin{corollary}\label{128}
Let $\mathcal{B}$ be a Boolean prealgebra. A subset $F$ is a filter of $\mathcal{B}$ if and only if the following conditions are satisfied for all $a,b\in B$:
\begin{itemize}
\item If $a,b\in F$, then $f_\wedge(a,b)\in F$.
\item If $a\in F$ and $a\preceq b$, then $b\in F$. 
\end{itemize}
A filter $F$ is a proper filter iff $F\neq B$ iff $f_\bot\notin F$. A filter $F$ is an ultrafilter iff $F$ is maximal among all proper filters. 
\end{corollary}

\begin{corollary}\label{120}
Let $\mathcal{B}$ be a Boolean prealgebra. If $\mathcal{B}$ is itself a Boolean algebra, then its lattice order $\le$ refines the given preorder $\preceq$, i.e., for all $a,b\in B$: $a\le b$ implies $a\preceq b$.
\end{corollary}

\begin{proof}
Suppose $\mathcal{B}$ is a Boolean algebra. Then for any $a,b\in B$: $a\le b$ $\Leftrightarrow$ $a=f^\mathcal{B}_\wedge(a,b)$ $\Rightarrow$ $\overline{a}=f^{\mathcal{B}/{\approx}}_\wedge(\overline{a},\overline{b})$ $\Leftrightarrow$ $\overline{a}\le^{\mathcal{B}/{\approx}}\overline{b}$ $\Leftrightarrow$ $a\preceq b$, where the last step follows from Lemma \ref{122}.
\end{proof}

\begin{definition}\label{200}
An $\mathit{SCI}$-model $\mathcal{M}$ is a structure
\begin{equation*}
\mathcal{M}=(M, \mathit{TRUE}, f_\vee, f_\wedge, f_\neg, f_\bot, f_\top, f_\rightarrow, f_\equiv, \preceq)
\end{equation*}
where $(M, f_\vee, f_\wedge, f_\neg, f_\bot, f_\top, f_\rightarrow, \preceq)$ is a Boolean prealgebra, $\mathit{TRUE}\subseteq M$ is a designated ultrafilter and $f_\equiv$ is an additional binary function satisfying for all $m,m'\in M$: $f_\equiv(m,m')\in\mathit{TRUE}\Leftrightarrow m = m'$. The elements of the universe $M$ are called propositions, and $\mathit{TRUE}$ is the designated set of true propositions.
\end{definition}

An assignment (or valuation) of an $\mathit{SCI}$-model $\mathcal{M}$ is a function $\gamma\colon V\rightarrow M$. Any assignment $\gamma$ extends in the canonical way to a function from $Fm_\equiv$ to $M$ which we again denote by $\gamma$. More precisely, we have $\gamma(\bot)=f_\bot$, $\gamma(\top)=f_\top$, and $\gamma(\varphi * \psi)=f_*(\gamma(\varphi),\gamma(\psi))$ for $*\in\{\vee, \wedge, \rightarrow, \equiv\}$. 

\begin{definition}\label{220}
If $\mathcal{M}$ is an $\mathit{SCI}$-model and $\gamma$ is an assignment of $\mathcal{M}$, then we call the tuple $(\mathcal{M},\gamma)$ an $\mathit{SCI}$-interpretation. The satisfaction relation between interpretations and formulas is defined as follows:
\begin{equation*}
(\mathcal{M},\gamma)\vDash\varphi :\Leftrightarrow \gamma(\varphi)\in \mathit{TRUE}
\end{equation*}
If $(\mathcal{M},\gamma)\vDash\varphi$ for all assignments $\gamma\in M^V$, then we write $\mathcal{M}\vDash\varphi$ and say that $\mathcal{M}$ validates $\varphi$ (or $\varphi$ is valid in $\mathcal{M}$). For $\varPhi\subseteq Fm_\equiv$, we define as usual $(\mathcal{M},\gamma)\vDash\varPhi :\Leftrightarrow (\mathcal{M},\gamma)\vDash\varphi \text{ for all }\varphi\in\varPhi$.
The relation of logical consequence is defined in the standard way for any set $\varPhi\cup\{\varphi\}\subseteq Fm_\equiv$: $\varPhi\Vdash_{\mathit{SCI}}\varphi :\Leftrightarrow Mod(\varPhi)\subseteq Mod(\{\varphi\})$, where for any $\varPsi\subseteq Fm_\equiv$, $Mod(\varPsi)$ is the class of all $\mathit{SCI}$-interpretations satisfying $\varPsi$.
\end{definition}

\begin{corollary}\label{230}
The connective of propositional identity has the intended meaning, i.e. for any interpretation $(\mathcal{M},\gamma)$ and any $\varphi,\psi\in Fm_\equiv$: $(\mathcal{M},\gamma)\vDash\varphi\equiv\psi$ iff $\gamma(\varphi)=\gamma(\psi)$ iff $\varphi$ and $\psi$ denote the same proposition in $(\mathcal{M},\gamma)$.
\end{corollary}

\begin{proof}
$(\mathcal{M},\gamma)\vDash\varphi\equiv\psi$ iff $\gamma(\varphi\equiv\psi)\in \mathit{TRUE}$ iff $f_\equiv(\gamma(\varphi),\gamma(\psi))\in\mathit{TRUE}$ iff $\gamma(\varphi)=\gamma(\psi)$.
\end{proof}

In \cite{blosus}, the authors consider only the logical connectives $\neg$ and $\rightarrow$, and consequently define an $\mathit{SCI}$-model as a structure $\mathcal{A}=(A,f_\neg,f_\rightarrow,f_\equiv)$ that satisfies certain conditions according to [Definition 1.6 \cite{blosus}] (we use here our specific notation for the semantic operations $f_\neg, f_\rightarrow, f_\equiv$ in order to keep the presentation consistent). By the following result, that original definition is essentially equivalent to our Definition \ref{200} of $\mathit{SCI}$-model presented above. This is not obvious since both definitions are formulated in very different ways. In particular, the original definition given in \cite{blosus} hides the prelattice structure which is an explicit part of our concept of $\mathit{SCI}$-model. 

\begin{theorem}[Equivalence of the two semantics]\label{240}
Our semantics based on Boolean prealgebras is equivalent to original semantics of $\mathit{SCI}$ in the following sense. Let $\mathcal{M}=(M, \mathit{TRUE}, f_\vee, f_\wedge, f_\neg, f_\bot, f_\top, f_\rightarrow, f_\equiv, \preceq)$ be an $\mathit{SCI}$-model according to Definition \ref{200}. Then the pair $<\mathcal{A},\mathit{TRUE}>$, where $\mathcal{A}=(M,f_\neg,f_\rightarrow,f_\equiv)$, is a model of $\mathit{SCI}$ according to [Definition 1.6 \cite{blosus}]. On the other hand, if $<\mathcal{A},B>$, with $\mathcal{A}=(A,f_\neg,f_\rightarrow,f_\equiv)$, is a model according to [Definition 1.6 \cite{blosus}], then $\mathcal{M}=(M, B, f_\vee, f_\wedge, f_\neg, f_\bot, f_\top, f_\rightarrow, f_\equiv, \preceq)$ is an $\mathit{SCI}$-model in our sense, where $a\preceq b\Leftrightarrow f_\rightarrow(a,b)\in B$, and the additional operations $f_\vee, f_\wedge, f_\bot, f_\top$ can be defined by the usual Boolean equations (e.g. $f_\top:=f_\rightarrow(a,a)$ for some fixed $a\in A$, etc.).
\end{theorem}

\begin{proof}
If $\mathcal{M}=(M, \mathit{TRUE}, f_\vee, f_\wedge, f_\neg, f_\bot, f_\top, f_\rightarrow, f_\equiv, \preceq)$ is an $\mathit{SCI}$-model in our sense, then, using the terminology of [Definition 1.6 \cite{blosus}], the set $\mathit{TRUE}$ is clearly closed, proper, prime and normal. Since $\mathcal{M}$ is based on a Boolean prealgebra, we have $f_\top\preceq h(\varphi)\in\mathit{TRUE}$ for any classical propositional tautology $\varphi$ and any valuation (assignment) $h$ of $\mathcal{M}$. Since $f_\equiv(a,b)\in\mathit{TRUE}\Leftrightarrow a=b$, also the identity axioms of $\mathit{SCI}$ are all interpreted by elements of $\mathit{TRUE}$ under any assignment $h$. Thus, $\mathit{TRUE}$ is also admissible and therefore a prime, normal filter according to [Definition 1.6 \cite{blosus}]. Thus, the reduct $\mathcal{A}=(M,f_\neg,f_\rightarrow,f_\equiv)$ along with prime, normal filter $\mathit{TRUE}\subseteq M$ yields an $\mathit{SCI}$-model in the original sense. Now let us suppose $<\mathcal{A},B>$, with $\mathcal{A}=(A,f_\neg,f_\rightarrow,f_\equiv)$, is a model in the sense of \cite{blosus}. We have to extract from that concept a preorder $\preceq$ that yields a prelattice and the desired $\mathit{SCI}$-model in the sense of Definition \ref{200} above. For elements $a,b\in A$, we define $a\preceq b :\Leftrightarrow$ $f_\rightarrow(a,b)\in B$. Since $B$ is a prime, normal filter (in the terminology of [Definition 1.6 \cite{blosus}], $B$ is, in a sense, deductively closed (i.e. if $h(\varPhi)\subseteq B$ and $\varPhi\vdash_{\mathit{SCI}}\varphi$, then $h(\varphi)\in B$, for any set $\varPhi\cup\{\varphi\}$ of formulas and any valuation $h$ of $\mathcal{A}$). It follows that $\preceq$ is a preorder on $A$, and $a\approx b :\Leftrightarrow$ ($a\preceq b$ and $b\preceq a$) defines a congruence relation of the structure $(A,f_\neg,f_\rightarrow)$. In particular, for any propositional formulas $\varphi,\psi$ (without identity connective), if $\varphi\leftrightarrow\psi$ is a theorem of $\mathit{CPC}$, then $h(\varphi)\approx h(\psi)$ under any valuation $h$. Thus, all Boolean equations are valid in the quotient structure of $(A,f_\neg, f_\rightarrow)$ modulo $\approx$, and that quotient structure must be a Boolean algebra (actually, it is the two-element Boolean algebra). We may define additional Boolean operations, such as $f_\wedge$ ... , in the obvious way. Furthermore, one easily verifies that the equivalence $a\preceq b\Leftrightarrow f_\wedge(a,b)\approx a$ is valid. Thus, $(A,f_\neg, f_\rightarrow,\preceq)$ is a Boolean prealgebra. Finally, the equivalence $f_\equiv(a,b)\in B\Leftrightarrow a=b$ is warranted by the fact that $B$ is normal (in the sense of [Definition 1.6 \cite{blosus}]). Thus, $\mathcal{M}=(A,B,f_\vee,f_\wedge,f_\neg,f_\bot,f_\top,f_\rightarrow,f_\equiv, \preceq)$ is an $\mathit{SCI}$-model according to Definition \ref{200} above. 
\end{proof}

One easily verifies that any $\mathit{SCI}$-interpretation (in our sense) satisfies the axioms of $\mathit{SCI}$. Completeness of $\mathit{SCI}$ w.r.t. our semantics follows from the original completeness theorem of $\mathit{SCI}$ (see, e.g. \cite{blosus}) together with Theorem \ref{240}. Nevertheless, we will sketch out in the following an independent proof. Suppose $\varPhi$ is a set of formulas which is consistent in $\mathit{SCI}$. By Zorn's Lemma, there is an extension $\varPsi\supseteq\varPhi$ which is maximal consistent in logic $\mathit{SCI}$. By the axioms of propositional identity, the relation $\cong$ defined by 
\begin{equation*}
\varphi\cong\psi :\Leftrightarrow (\varphi\equiv\psi)\in\varPsi
\end{equation*}
is a congruence relation on $Fm_\equiv$ (symmetry, transitivity and compatibility with operations follow from applications of \eqref{10}, i.e. the Substitution Property SP). Moreover, by (id2), $\varphi\cong\psi$ implies: $\varphi\in\varPsi\Leftrightarrow\psi\in\varPsi$. For $\varphi\in Fm_\equiv$, let $[\varphi]$ be the congruence class of $\varphi$ modulo $\cong$. Then we put $M:=\{[\varphi]\mid\varphi\in Fm_\equiv\}$, $\mathit{TRUE}:=\{[\varphi]\mid\varphi\in\varPsi\}$ and define operations $f_\neg([\varphi]):=[\neg\varphi]$, $f_*([\varphi],[\psi]):=[\varphi * \psi]$, for $*\in\{\vee, \wedge, \rightarrow, \equiv\}$, and $f_\bot:=[\bot]$, $f_\top:=[\top]$. The relation $\preceq$ on $M$ defined by 
\begin{equation*}
[\varphi]\preceq [\psi] :\Leftrightarrow\varphi\rightarrow\psi\in\varPsi
\end{equation*}
is a preorder on $M$. By SP, $\preceq$ is well-defined. Next we show that the structure
\begin{equation*}
\mathcal{M'}=(M, f_\vee, f_\wedge, f_\neg, f_\bot, f_\top, f_\rightarrow, \preceq)
\end{equation*}
is a Boolean prealgebra. The relation $\approx$ given by 
\begin{equation*}
[\varphi]\approx [\psi] :\Leftrightarrow \varphi\leftrightarrow\psi\in\varPsi \Leftrightarrow ([\varphi]\preceq [\psi]\text{ and }[\psi]\preceq [\varphi])
\end{equation*}
is obviously a congruence relation of $\mathcal{M'}$. Since $\varPsi$ is maximal consistent, it contains in particular all equivalences $\varphi\leftrightarrow\psi$ which are valid in $\mathit{CPC}$. These equivalences axiomatize as equations `$\varphi =\psi$' the class of Boolean algebras. It follows that the quotient of $\mathcal{M'}$ modulo $\approx$ is a Boolean algebra whose elements are the congruence classes of the elements $[\varphi]\in M$ modulo $\approx$. Moreover, for any elements $[\varphi],[\psi]$ we have: $[\varphi]\preceq [\psi]$ iff $\varphi\rightarrow\psi\in\varPsi$ iff $(\varphi\wedge\psi)\leftrightarrow\varphi\in\varPsi$ iff $f_\wedge([\varphi],[\psi])\approx [\varphi]$. Hence, $\mathcal{M}'$ is a Boolean prealgebra in accordance with Definition \ref{100}.\footnote{$\mathcal{M'}$ is not necessarily a Boolean algebra. For example, $f_\vee ([\varphi],[\psi])=[\varphi\vee\psi]\neq [\psi\vee\varphi]=f_\vee([\psi], [\varphi])$ is possible. Even if $\mathcal{M'}$ is a Boolean algebra, the preorder $\preceq$ may be strictly coarser than the underlying lattice order (cf. Lemma \ref{120}). In fact, $\preceq$ is the lattice order iff $\varPsi$ contains all instances of the Fregean Axiom $(\varphi\equiv\psi)\leftrightarrow (\varphi\leftrightarrow\psi)$.} By construction, we have for any elements $[\varphi]$, $[\psi]$: $[\varphi] = [\psi]$ iff $\varphi\cong\psi$ iff $\varphi\equiv\psi\in\varPsi$ iff $[\varphi\equiv\psi]=f_\equiv([\varphi],[\psi])\in\mathit{TRUE}$. Thus, 
\begin{equation*}
\mathcal{M}:=(M, \mathit{TRUE}, f_\vee, f_\wedge, f_\neg, f_\bot, f_\top, f_\rightarrow, f_\equiv, \preceq)
\end{equation*}
is an $\mathit{SCI}$-model. We consider the assignment $\gamma\in M^V$ defined by $x\mapsto [x]$. By induction on formulas, it follows that $\gamma(\varphi)=[\varphi]$. Then we have 
\begin{equation*}
(\mathcal{M},\gamma)\vDash\varphi\Leftrightarrow\gamma(\varphi)=[\varphi]\in\mathit{TRUE}\Leftrightarrow\varphi\in\varPsi.
\end{equation*}
In particular, $(\mathcal{M},\gamma)\vDash\varPhi$ and whence $\varPhi$ is satisfiable. We have proved soundness and completeness of $\mathit{SCI}$ w.r.t. the semantics given by the class of $\mathit{SCI}$-models.

\begin{theorem}[Soundness and Completeness]\label{250}
For any set $\varPhi\cup\{\varphi\}\subseteq Fm_\equiv$, the following holds: $\varPhi\Vdash_{\mathit{SCI}}\varphi \Leftrightarrow \varPhi\vdash_{\mathit{SCI}}\varphi$.
\end{theorem}

Classical propositional logic $\mathit{CPC}$ is extensional in the sense that the denotation (reference, \textit{Bedeutung}) of any formula is given by its truth-value relative to the underlying assignment: either true or false. Consequently, the Fregean Axiom holds: $(\varphi\leftrightarrow\psi)\leftrightarrow (\varphi\equiv\psi)$. It is known that this situation can be modeled in $\mathit{SCI}$ by presenting a two-element model where all true formulas denote one element (the true proposition) and all false formulas denote the other one (the false proposition).

\begin{example}\label{280}
There exists an extensional $\mathit{SCI}$-model, i.e. a two-element model $\mathcal{M}$ where the denotation of a formula is nothing but a classical truth value: for every assignment $\gamma\colon V\rightarrow\{0,1\}$ and all $\varphi,\psi\in Fm$, $(\mathcal{M},\gamma)\vDash\varphi\equiv\psi$ iff $(\mathcal{M},\gamma)\vDash\varphi\leftrightarrow\psi$ iff $\varphi$ and $\psi$ have the same classical truth-value.\\
Of course, the desired extensional model will be based (up to isomorphism) on the two-element Boolean algebra $\mathcal{B}$ with universe $\{0,1\}$. Let $\preceq$ be the natural total order on $\{0,1\}$. The resulting relation $\approx$ is the identity and the associated quotient algebra is $\mathcal{B}$ itself. We define an additional Boolean operation $f_\equiv\colon\{0,1\}\times\{0,1\}\rightarrow\{0,1\}$ by $f_\equiv(x,y)=1$ $:\Leftrightarrow$ $x=y$. Then $\mathcal{B}$ together with $f_\equiv$ and the unique ultrafilter $\mathit{TRUE}:=\{1\}$ yields an $\mathit{SCI}$-model $\mathcal{M}$. Obviously, for any assignment $\gamma\colon V\rightarrow\{0,1\}$ and for any formulas $\varphi,\psi\in F_\equiv$, we have $(\mathcal{M},\gamma)\vDash\varphi\equiv\psi$ iff $\gamma(\varphi)=\gamma(\psi)$ iff $\varphi$ and $\psi$ have the same classical truth-value.
\end{example}

It is clear that the above two-valued $\mathit{SCI}$-model along with all possible assignments yields essentially the standard two-valued semantics of classical propositional logic $\mathit{CPC}$. In fact, $\mathit{CPC}$ is represented by the $\mathit{SCI}$-theory $\mathit{SCI}^{ext}$ that results from $\mathit{SCI}$ by adding Fregean Axiom $(\varphi\leftrightarrow\psi)\rightarrow (\varphi\equiv\psi)$. Theory $\mathit{SCI}^{ext}$ contains $(\varphi\leftrightarrow\psi)\leftrightarrow (\varphi\equiv\psi)$ and thus $(\varphi\leftrightarrow\psi)\equiv (\varphi\equiv\psi)$ as theorems. By SP, $(\varphi\leftrightarrow\psi)$ and $(\varphi\equiv\psi)$ then can be replaced by each other in every context. One easily shows that $\mathit{SCI}^{ext}$ is sound and complete w.r.t. the class of all extensional (i.e., two-element) $\mathit{SCI}$-models. We have for any $\varphi\in Fm_\equiv$:
\begin{equation*}
\vdash_{\mathit{SCI}^{ext}}\varphi\text{ }\Leftrightarrow\text{ } \vdash_{\mathit{CPC}}\varphi^*,
\end{equation*}
where $\varphi^*$ is the result of replacing every subformula of the form $\psi\equiv\chi$ in $\varphi$ by $\psi\leftrightarrow\chi$. \\

Another important example of $\mathit{SCI}$-model, as opposed to an extensional model, is an intensional model where the denotation of a formula is determined by its intension, i.e. its syntactical form. In such a model, any two (syntactically) different formulas have different denotations. The denotation of a formula can be identified with its intension. In the following, we present a construction of such a model. Intensional models have also been constructed for a logic that extends $\mathit{SCI}$ by propositional quantifiers and a truth predicate (see, e.g. the discussion and a construction presented in \cite{lewigpl}).\footnote{The construction of an intensional model for such a first-order logic is not trivial because of the impredicativity of propositional quantifiers. Note that bound variable $x$ in formula $\forall x\varphi$ ranges over the universe of all propositions which contains in particular the proposition denoted by $\forall x\varphi$ itself.} 

\begin{example}\label{300}
There exists an intensional $\mathit{SCI}$-model, i.a. a model $\mathcal{M}$ along with an assignment $\gamma$ such that for all $\varphi,\psi\in Fm_\equiv$,
\begin{equation*}
(\mathcal{M},\gamma)\vDash\varphi\equiv\psi\Leftrightarrow\varphi=\psi.
\end{equation*}

Let us construct model $\mathcal{M}$. We define a rank $R\colon Fm_\equiv\rightarrow\mathbb{N}$ on formulas as follows:
\begin{itemize}
\item $R(x)=R(\bot)=R(\top)=R(\varphi\equiv\psi)=0$, for any $x\in V$ and $\varphi,\psi\in Fm_\equiv$.
\item If $\varphi,\psi\in Fm$ such that $R(\varphi)$ and $R(\psi)$ are already defined, then $R(\neg\varphi)=R(\varphi)+1$ and $R(\varphi*\psi)=max\{R(\varphi), R(\psi)\}+1$, where $*\in\{\vee, \wedge,\rightarrow\}$. 
\end{itemize}
We consider the given enumeration of the set of variables $V=\{x_0, x_1, x_2, ... \}$ and define the set $\mathit{TRUE}$ by induction on rank $R$ as the smallest set such that the following conditions are satisfied:
\begin{itemize}
\item For formulas of rank $0$, we have: $\bot\not\in\mathit{TRUE}$, $\top\in\mathit{TRUE}$, $x_i\in\mathit{TRUE}$ iff $i$ is an even index, $\varphi\equiv\psi\in\mathit{TRUE}$ iff $\varphi=\psi$.
\item Suppose membership of all formulas of rank $\le n\in\mathbb{N}$ w.r.t. $\mathit{TRUE}$ is already determined. Let $\varphi$, $\psi$ be formulas such that $max\{R(\varphi), R(\psi)\} = n$. Then:
\begin{itemize}
\item $\varphi\wedge\psi\in\mathit{TRUE}$ if $\varphi\in\mathit{TRUE}$ and $\psi\in\mathit{TRUE}$
\item $\varphi\vee\psi\in\mathit{TRUE}$ if $\varphi\in\mathit{TRUE}$ or $\psi\in\mathit{TRUE}$
\item $\neg\varphi\in\mathit{TRUE}$ if $\varphi\notin\mathit{TRUE}$
\item $\varphi\rightarrow\psi\in\mathit{TRUE}$ if $\varphi\notin\mathit{TRUE}$ or $\psi\in\mathit{TRUE}$
\end{itemize}
\end{itemize}
Membership w.r.t. $\mathit{TRUE}$ determines a classical truth-value for every formula. The relation $\preceq$ on $M:=Fm_\equiv$ defined by $\varphi\preceq\psi :\Leftrightarrow$ $\varphi\rightarrow\psi\in\mathit{TRUE}$ is a preorder. Moreover, the relation $\approx$ defined by $\varphi\approx\psi$ $:\Leftrightarrow$ ($\varphi\preceq\psi$ and $\psi\preceq\varphi$) $\Leftrightarrow$ `both $\varphi$ and $\psi$ belong to $\mathit{TRUE}$ or both $\varphi$ and $\psi$ belong to $M\smallsetminus\mathit{TRUE}$' is a congruence relation on the structure $(M,\vee,\wedge,\neg,\bot,\top,\rightarrow)$. The associated quotient algebra is the two-element Boolean algebra $\{0,1\}$ where $1$ is the image of $\mathit{TRUE}$ under the canonical homomorphism. Moreover, $\varphi\preceq\psi$ $\Leftrightarrow$ $\varphi\rightarrow\psi\in\mathit{TRUE}$ $\Leftrightarrow$ [$(\varphi\wedge\psi)\rightarrow\varphi\in \mathit{TRUE}$ and $\varphi\rightarrow (\varphi\wedge\psi)\in\mathit{TRUE}$] $\Leftrightarrow$ $(\varphi\wedge\psi)\approx\varphi$. Hence,
\begin{equation*}
\mathcal{M'}=(M,\vee,\wedge,\neg,\bot,\top,\rightarrow,\preceq)
\end{equation*}
is a Boolean prealgebra. Together with ultrafilter $\mathit{TRUE}$ and the operation $f_\equiv$ on $M=Fm_\equiv$ defined by $f_\equiv(\varphi,\psi):=(\varphi\equiv\psi)$ we then obtain the $\mathit{SCI}$-model 
\begin{equation*}
\mathcal{M}=(M,\mathit{TRUE},\vee, \wedge,\neg,\bot,\top,\rightarrow,f_\equiv,\preceq).
\end{equation*}
Consider the assignment $\gamma\colon V\rightarrow Fm_\equiv$, $x\mapsto x$. Then, by induction on formulas, $\gamma(\varphi)=\varphi$ for any $\varphi\in Fm_\equiv$. Furthermore, for all $\varphi,\psi\in Fm_\equiv$: 
\begin{equation*}
\begin{split}
(\mathcal{M},\gamma)\vDash\varphi\equiv\psi & \Leftrightarrow \gamma(\varphi\equiv\psi)=f_\equiv(\gamma(\varphi),\gamma(\psi))\in\mathit{TRUE}\\
& \Leftrightarrow\gamma(\varphi)=\gamma(\psi)\Leftrightarrow\varphi =\psi.
\end{split}
\end{equation*}
\end{example}

As a consequence, already observed by Suszko, only trivial identities are theorems of $\mathit{SCI}$.

\begin{corollary}\label{320}
For all $\varphi,\psi\in Fm_\equiv$, $\vdash_{\mathit{SCI}}\varphi\equiv\psi\Leftrightarrow\varphi=\psi$.
\end{corollary}

\begin{proof}
If $\varphi=\psi$, then by identity axiom (id1): $\vdash_{\mathit{SCI}}\varphi\equiv\psi$. On the other hand, if $\varphi\neq\psi$, then we have $(\mathcal{M},\gamma)\nvDash\varphi\equiv\psi$ for the intensional model constructed above and thus $\varphi\equiv\psi$ is not logically valid. Soundness yields $\nvdash_{\mathit{SCI}}\varphi\equiv\psi$. 
\end{proof}

In the remainder of this section, we show that some relevant modal principles can be restored in the pure $\mathit{SCI}$, i.e. in $\mathit{SCI}$ with no additional axioms. The representation of certain Lewis-style modal systems by means of appropriate $\mathit{SCI}$-extensions will be the topic of the next section.\\

For $\varphi\in Fm_\equiv$, we define 
\begin{equation}\label{330}
\square\varphi := (\varphi\equiv\top).
\end{equation}

\begin{theorem}\label{340}
Let $\mathcal{M}$ be an $\mathit{SCI}$-model. Then the following are equivalent:
\begin{enumerate}
\item $\mathcal{M}$ is based on a Boolean algebra, i.e. its $\{f_\vee,f_\wedge,f_\neg,f_\bot,f_\top\}$-reduct is a Boolean algebra.
\item For all formulas $\chi$ having the form of a classical tautology, and for all formulas $\varphi$ and $\psi$, model $\mathcal{M}$ validates $\square\chi$ and $(\varphi\equiv\psi)\leftrightarrow \square(\varphi\leftrightarrow\psi)$.
\end{enumerate}
\end{theorem}

\begin{proof}
If $\mathcal{M}$ is a Boolean algebra, then all theorems of $\mathit{CPC}$, as well as their substitution instances, are evaluated by the top element $f_\top$ under any assignment. It is also known that the equivalence $f_\rightarrow(m,m')=f_\top$ $\Leftrightarrow$ $m\le m'$ holds in every Boolean algebra (actually, in every Heyting algebra). Then it is clear that (i) implies (ii). Now, suppose (ii) holds true. Then $\mathcal{M}$ validates in particular $\square(\varphi\leftrightarrow\psi)$ whenever $\varphi\leftrightarrow\psi$ is a classical tautology. Since $(\varphi\equiv\psi)\leftrightarrow \square(\varphi\leftrightarrow\psi)$ is valid in $\mathcal{M}$, we have $\mathcal{M}\vDash \varphi\equiv\psi$ for all Boolean equations $\varphi\equiv\psi$ that axiomatize the class of Boolean algebras. Hence, $\mathcal{M}$ itself is based on a Boolean algebra. 
\end{proof}

\begin{definition}\label{350}
$\mathit{SCI}^+$ is the logic that results from $\mathit{SCI}$ by adding the following axioms:
\begin{itemize}
\item $\square\chi$ whenever $\chi$ has the form of a classical tautology,
\item $(\varphi\equiv\psi)\leftrightarrow \square(\varphi\leftrightarrow\psi)$.
\end{itemize}
\end{definition}

The next result then follows from Theorem \ref{340}.

\begin{corollary}\label{360}
The $\mathit{SCI}$-extension $\mathit{SCI}^+$ is sound and complete w.r.t. the class of those $\mathit{SCI}$-models which are based on Boolean algebras. As a consequence, $\mathit{SCI}^+$ coincides with the known $\mathit{SCI}$-theory $\mathit{WB}$.\footnote{Theory $\mathit{WB}$ is discussed in some works on non-Fregean logic (see, e.g. \cite{waw} for a detailed presentation).}
\end{corollary}

The question arises whether theory $\mathit{SCI}^+$ contains further interesting modal laws. Using Corollary \ref{360}, we may argue semantically showing that the following formulas are theorems of $\mathit{SCI^+}$:
\begin{itemize}
\item $\square\varphi\rightarrow\varphi$
\item $\square(\varphi\rightarrow\psi)\rightarrow(\square\varphi\rightarrow\square\psi)$.
\end{itemize}
In fact, given a Boolean algebra, the top element $f_\top$ is contained in every ultrafilter; and for any elements $m,m'$: if $m\le m'$, then $m=f_\top$ implies $m'=f_\top$. Thus, the validity of the above formulas is justified. However, some principles of normal Lewis systems are not valid. For instance, the full necessitation rule does not hold. As a contra-example, we consider the Boolean algebra $2^2$ with elements $\varnothing$, $\{0\}$, $\{1\}$, $\{0,1\}$ and set-theoretic inclusion as lattice order, along with the ultrafilter $\mathit{TRUE}=\{\{0\},\{0,1\}\}$ and operation $f_\equiv$ defined by $f_\equiv(m,m'):=\{0\}\in\mathit{TRUE}$ if $m=m'$, and $f_\equiv(m,m')=\{1\}\notin\mathit{TRUE}$ otherwise. Then $f_\square(\{0\})=f_\equiv(\{0\},f_\top)=\{1\}\not\le\{0\}$. Thus, $\square (\square\varphi\rightarrow\varphi)$ is not valid.

\begin{definition}\label{370}
In some analogy to Lewis modal system $\mathit{S3}$, we define the following extension of $\mathit{SCI}^+$: $\mathit{SCI}_3$ is the logic that results from $\mathit{SCI}^+$ by adding all formulas of the form $\square(\varphi\rightarrow\psi)\rightarrow\square(\square\varphi\rightarrow\square\psi)$ as theorems. 
\end{definition}

Note, however, that the alleged analogy to Lewis system $\mathit{S3}$ is rather weak. For instance, $\square (\square(\varphi\rightarrow\psi)\rightarrow\square(\square\varphi\rightarrow\square\psi))$ is a theorem of $\mathit{S3}$ but not of $\mathit{SCI}_3$.

\begin{definition}\label{380}
An $\mathit{SCI}$-model $\mathcal{M}$ is an $\mathit{SCI}_3$-model if $\mathcal{M}$ is based on a Boolean algebra and satisfies the following condition for all $m,m'\in M$:
\begin{equation*}
m\le m' \Rightarrow f_\square(m)\le f_\square(m'),
\end{equation*}
where $f_\square(m):=f_\equiv(m,f_\top)$. That is, $f_\square$ is monotonic on $M$.
\end{definition}

\begin{corollary}\label{390}
Logic $\mathit{SCI}_3$ is sound and complete w.r.t. the class of $\mathit{SCI}_3$-models.
\end{corollary}

\begin{proof}
One easily checks that every $\mathit{SCI}_3$-model validates formulas of the form $\square(\varphi\rightarrow\psi)\rightarrow\square(\square\varphi\rightarrow\square\psi)$. In order to prove completeness, it is enough to show that the constructed model in the proof of Theorem \ref{250} above satisfies the condition of monotonicity of $f_\square$. Since $\mathit{SCI}_3$ contains $\mathit{SCI}^+$, we already know that that model is a Boolean algebra. So for two elements $[\varphi]$ and $[\psi]$, suppose $[\varphi]\le [\psi]$ (where $\le$ is the lattice order). Then $(\varphi\rightarrow\psi) \equiv\top\in\varPsi$. That is, $\square(\varphi\rightarrow\psi)\in \varPsi$ and thus $\square(\square\varphi\rightarrow\square\psi)\in\varPsi$. But then $(\square\varphi\rightarrow\square\psi)\equiv\top\in\varPsi$ and thus $[\square\varphi\rightarrow\square\psi]=[\top]$, i.e. $f_\square([\varphi])=[\square\varphi]\le [\square\psi]=f_\square([\psi])$.
\end{proof}

We are interested in conditions that ensure, in some precise sense, complete restorations of some Lewis-style modal systems, in particular of $\mathit{S3}$--$\mathit{S5}$. It turns out that principle $(\varphi\equiv\psi)\leftrightarrow \square(\varphi\leftrightarrow\psi)$, valid in $\mathit{SCI}^+$, is too weak for this purpose. In fact, we must postulate the equation $(\varphi\equiv\psi)\equiv\square(\varphi\leftrightarrow\psi)$, i.e. we must identify propositional identity with strict equivalence. These topics will be studied in section 5.

\section{Some Lewis-style modal systems and their algebraic semantics}

The goal of this section is to revise some Lewis-style modal systems in the vicinity of $\mathit{S3}$ (more precisely, systems based on a logic called $\mathit{S1SP}$) which in the subsequent section then will be shown to be dual, in some precise sense, to certain $\mathit{SCI}$-theories. Our object language is now the language of propositional modal logic $Fm_\square$, i.e. the set of formulas inductively defined over the set of variables $V=\{x_0,x_1,...\}$, logical connectives $\bot, \top, \vee, \wedge, \neg, \rightarrow$ and the modal operator $\square$. Thus, the languages $Fm_\equiv$ and $Fm_\square$ share the `pure' propositional part based on the logical connectives. We introduce an `identity connective' defined by strict equivalence:
\begin{equation}\label{405}
(\varphi\equiv\psi) := (\square(\varphi\rightarrow\psi)\wedge\square(\psi\rightarrow\varphi)).
\end{equation}
It is evident that under this interpretation, all Lewis modal systems $\mathit{S1}$--$\mathit{S5}$ satisfy Suszko's identity axioms (id1) $\varphi\equiv\varphi$ and (id2) $(\varphi\equiv\psi)\rightarrow (\varphi\rightarrow\psi)$. Moreover, $\mathit{S3}$ also satisfies the remaining identity axioms, i.e. SP ( where, of course, identity is given as strict equivalence according to \eqref{405} above). $\mathit{S3}$ is the weakest Lewis modal system containing SP (cf. \cite{lewjlc1, lewsl}). In the following, we recall definitions of some relevant Lewis-style modal systems and consider an algebraic semantics which can be immediately translated into $\mathit{SCI}$-semantics, and vice-versa. We adopt that particular approach to algebraic semantics from \cite{lewjlc1}. \\

Lewis system $\mathit{S1}$ can be defined in the following way (cf., e.g., \cite{hugcre}). All formulas of the following form are axioms:
\begin{itemize}
\item tautologies (and their substitution-instances) of $\mathit{CPC}$
\item $\square\varphi\rightarrow\varphi$
\item $(\square(\varphi\rightarrow\psi)\wedge \square(\psi\rightarrow\chi))\rightarrow\square(\varphi\rightarrow\chi)$ (transitivity of strict implication)
\end{itemize}
The inference rules are Modus Ponens MP, Axiom Necessitation AN ``If $\varphi$ is an axiom, then $\square\varphi$ is a theorem", and Substitution of Proved Strict Equivalents SPSE ``If $\varphi\equiv\psi$ is a theorem, then so is $\chi[x:=\varphi]\equiv\chi[x:=\psi]$". \\

Lewis system $\mathit{S3}$ results from $\mathit{S1}$ by adding\\

(S3) $\square(\varphi\rightarrow\psi)\rightarrow\square(\square\varphi\rightarrow\square\psi)$\\

as an axiom scheme to $\mathit{S1}$. Of course, rule (AN) now applies also to (S3). Rule SPSE can be ignored since it is derivable from the rest. \\
Lewis system $\mathit{S4}$ results from $\mathit{S3}$ by adding \\

(S4) $\square\varphi\rightarrow\square\square\varphi$\\

as an axiom scheme (rule (AN) now applies also to (S4)). Finally, $\mathit{S5}$ results from $\mathit{S4}$ by adding\\

(S5) $\neg\square\varphi\rightarrow\square\neg\square\varphi$\\

as an axiom scheme. 

We do not consider Lewis system $\mathit{S2}$ since it is apparently not susceptible to our algebraic semantics. Recall, however, that $\mathit{S2}$ can be captured by a non-normal Kripke-style semantics. There is no known natural semantics for Lewis system $\mathit{S1}$ (cf. \cite{hugcre}). If we strengthen the $\mathit{S1}$-rule SPSE to our stronger Substitution Principle SP, $(\varphi\equiv\psi)\rightarrow (\chi[x:=\varphi]\equiv\chi[x:=\psi])$, and add it as a theorem scheme to $\mathit{S1}$ (i.e., SP is regarded a scheme of theorems; recall that AN is not applicable to theorems), then we obtain modal system $\mathit{S1+SP}$ which was introduced and studied in \cite{lewjlc1}. Simplifying notation, we will refer to that system as $\mathit{S1SP}$ instead of $\mathit{S1+SP}$. In contrast to $\mathit{S1}$, the stronger system $\mathit{S1SP}$ has a natural model-theoretic semantics which we will recall below.

In system $\mathit{S1}$, derivations from the empty set, i.e. derivations of theorems, are defined as usual. For $\mathcal{L}\in\{\mathit{S1SP}, \mathit{S3}, \mathit{S4}, \mathit{S5}\}$ and $\varPhi\cup\{\varphi\}\subseteq Fm_\square$, we write $\varPhi\vdash_\mathcal{L}\varphi$ if there is a derivation of $\varphi$ from $\varPhi$, i.e. a finite sequence $\varphi_1,...,\varphi_n=\varphi$ such that for each $\varphi_i$, $i\le i\le n$, the following holds: $\varphi_i\in\varPhi$ or $\varphi_i$ is an axiom of $\mathcal{L}$ or $\varphi_i$ is obtained by AN (i.e. $\varphi_i = \square\psi$ for some axiom $\psi$ of $\mathcal{L}$) or $\varphi_i$ is obtained by MP applied to preceding formulas of the sequence. Note that we can do without the full Necessitation Rule ``If $\varphi$ is a theorem, then so is $\square\varphi$". In fact, by induction on derivations one shows that the full Necessitation Rule is derivable in $\mathit{S4}$.

The following result is proven in [\cite{lewjlc1}, Lemma 2.3] where it is originally formulated for logic $\mathit{S1SP}$. The proof given there makes use of SP. However, one recognizes that SP can be replaced by the $\mathit{S1}$-rule SPSE in the proof. Hence, the result also holds in the weaker system $\mathit{S1}$.

\begin{lemma}[\cite{lewjlc1}]\label{410}
Every instance of the following principle N is a theorem of $\mathit{S1}$:
\begin{equation*}
\square\varphi\leftrightarrow (\varphi\equiv\top).
\end{equation*}
\end{lemma}

N expresses the fact that there exists exactly one necessary proposition, namely the proposition denoted by $\top$.\\

N would easily follow from distribution principle K, $\square(\varphi\rightarrow\psi)\rightarrow(\square\varphi\rightarrow\square\psi)$.\footnote{Consider classical tautology $\varphi\leftrightarrow (\varphi\leftrightarrow\top)$, rule AN, principle K and MP.} However, K is not available in $\mathit{S1}$. Nevertheless, using N and SP we are able to show the following (cf. \cite{lewjlc1}, Lemma 2.4):

\begin{lemma}[\cite{lewjlc1}]\label{420}
Distribution principle K holds in $\mathit{S1SP}$, i.e. formulas of the form 
\begin{equation*}
\square(\varphi\rightarrow\psi)\rightarrow(\square\varphi\rightarrow\square\psi)
\end{equation*}
are theorems of $\mathit{S1SP}$.
\end{lemma}

\begin{lemma}\label{425}
Equivalences $\square(\varphi\wedge\psi)\leftrightarrow (\square\varphi\wedge\square\psi)$ are theorems of $\mathit{S1SP}$.
\end{lemma}

\begin{proof}
We show that $\square(\varphi\wedge\psi)\rightarrow (\square\varphi\wedge\square\psi)$ is a theorem. By Lemma \ref{410}, $\square(\varphi\wedge\psi)\leftrightarrow ((\varphi\wedge\psi)\equiv\top)$. In particular, we have the following valid implication: $\square(\varphi\wedge\psi)\rightarrow \square(\top\rightarrow (\varphi\wedge\psi))$. By the transitivity axiom of strict implication of $\mathit{S1}$, $(\square(\top\rightarrow (\varphi\wedge\psi))\wedge\square((\varphi\wedge\psi)\rightarrow\varphi))\rightarrow\square(\top\rightarrow\varphi)$. Note that $\square((\varphi\wedge\psi)\rightarrow\varphi))$ results from an application of rule AN. Then transitivity of implication yields $\square(\varphi\wedge\psi)\rightarrow\square(\top\rightarrow\varphi)$, i.e. $\square(\varphi\wedge\psi)\rightarrow (\varphi\equiv\top)$. By principle N, we get $\square(\varphi\wedge\psi)\rightarrow\square\varphi$. Similarly, we get $\square(\varphi\wedge\psi)\rightarrow\square\psi$ and thus $\square(\varphi\wedge\psi)\rightarrow (\square\varphi\wedge\square\psi)$.\\
Now, we show the converse $(\square\varphi\wedge\square\psi)\rightarrow \square(\varphi\wedge\psi)$ making use of SP. Note that $\square\psi\leftrightarrow (\psi\equiv\top)$ and $(\psi\equiv\top)\rightarrow\square (\varphi\wedge y)[y:=\psi]\equiv \square(\varphi\wedge y)[y:=\top]$ are instances of N and SP, respectively. By transitivity of implication, $\square\psi\rightarrow (\square (\varphi\wedge\psi)\equiv \square(\varphi\wedge \top))$ is a theorem. Thus, $\square\psi\rightarrow (\square(\varphi\wedge \top)\rightarrow \square (\varphi\wedge\psi))$ is a theorem. By rule AN, $\varphi\equiv (\varphi\wedge\top)$ is a theorem. Then we may apply $\mathit{S1}$-rule SPSE (or the stronger SP) and derive $\square\psi\rightarrow (\square\varphi\rightarrow \square (\varphi\wedge\psi))$ which modulo $\mathit{CPC}$ is equivalent to $(\square\varphi\wedge\square\psi)\rightarrow \square(\varphi\wedge\psi)$.
\end{proof}

By Lemma \ref{425}, we may write strict equivalence $\square(\varphi\rightarrow\psi)\wedge \square(\psi\rightarrow\varphi)$ equivalently and shorter as $\square(\varphi\leftrightarrow\psi)$ in systems containing $\mathit{S1SP}$. In $\mathit{S1SP}$, we may also strengthen the result of Lemma \ref{410} as follows.

\begin{lemma}\label{430}
The following scheme $\square N$ is derivable in $\mathit{S1SP}$:
\begin{equation*}
\square\varphi\equiv (\varphi\equiv\top).
\end{equation*}
\end{lemma}

\begin{proof}
Note that $\varphi\leftrightarrow (\varphi\leftrightarrow\top)$ is a propositional tautology. Rule AN yields $\Box(\varphi\leftrightarrow (\varphi\leftrightarrow\top))$, i.e. $\varphi\equiv (\varphi\leftrightarrow\top)$. Consider the instance $(\varphi\equiv (\varphi\leftrightarrow\top))\rightarrow (\square x[x:=\varphi]\equiv \square x[x:=(\varphi\leftrightarrow\top)]$ of SP and apply MP. This yields theorem $\square\varphi\equiv \square(\varphi\leftrightarrow\top)$
\end{proof}

In the following definitions, by a Boolean algebra expansion we always mean a structure $\mathcal{M}=(M, \mathit{TRUE}, f_\vee, f_\wedge, f_\neg, f_\bot, f_\top, f_\rightarrow, f_\square)$ which is based on a Boolean algebra with the usual operations along with a designated ultrafilter $\mathit{TRUE}$ and an additional unary function $f_\square$. The induced lattice order is always denoted by $\le$.

\begin{definition}\label{570}
Let $\mathcal{M}$ be a Boolean algebra expansion satisfying the following conditions for all $a,b,c\in M$:\\
(1) $f_\square(a)\in\mathit{TRUE}\Leftrightarrow a=f_\top$\\
(2) $f_\square(a)\le a$\\
(3) $f_\wedge(f_\square(f_\rightarrow(a,b)),f_\square(f_\rightarrow(b,c)))\le f_\square(f_\rightarrow(a,c))$\\
Then we call $\mathcal{M}$ an $\mathit{S1SP}$-algebra.
\end{definition}

Note that conditions (2) and (3) reflect corresponding axioms of $\mathit{S1}$.

\begin{lemma}\label{580}
In every $\mathit{S1SP}$-algebra it holds that 
\begin{equation*}
f_\square(f_\wedge(a,b))\in\mathit{TRUE} \Leftrightarrow f_\wedge(f_\square(a), f_\square(b))\in\mathit{TRUE},
\end{equation*}
for all elements $a,b$, i.e. formulas of the form 
\begin{equation*}
\square(\varphi\wedge\psi)\leftrightarrow (\square\varphi\wedge\square\psi)
\end{equation*}
are valid in the class of $\mathit{S1SP}$-algebras. Moreover, modal principle $K$, 
\begin{equation*}
\square(\varphi\rightarrow\psi)\rightarrow (\square\varphi\rightarrow\square\psi),
\end{equation*} 
is valid in the class of $\mathit{S1SP}$-algebras.
\end{lemma}

\begin{proof}
By (1), $f_\square(f_\wedge(a,b))\in\mathit{TRUE}$ $\Leftrightarrow$ $a=f_\top$ and $b=f_\top$ $\Leftrightarrow$ $f_\square(a)\in\mathit{TRUE}$ and $f_\square(b)\in\mathit{TRUE}$ $\Leftrightarrow$ $f_\wedge(f_\square(a), f_\square(b))\in\mathit{TRUE}$. The second assertion can be shown as follows: For a given $\mathit{S1SP}$-algebra, suppose $f_\square(f_\rightarrow(a,b))\in\mathit{TRUE}$ and $f_\square(a)\in\mathit{TRUE}$. The former implies $f_\rightarrow(a,b)=f_\top$, i.e. $a\le b$. The latter implies $a=f_\top$. It follows $b=f_\top$ and thus $f_\square(b)\in\mathit{TRUE}$.
\end{proof}

Notice that validity of modal principle $K$ in the class of $\mathit{S1SP}$-algebras does not mean that all instances of $K$ are interpreted by the top element of the given Boolean algebra (as it is the case in normal modal logics). It only means that such instances are interpreted by some element of the ultrafilter $\mathit{TRUE}$, a designated ultrafilter that contains in particular the element $f_\square(f_\top)$. In fact, we cannot choose an arbitrary ultrafilter $\mathit{TRUE}$ of the Boolean algebra: condition (1) of Definition \ref{570} must be fulfilled. In this aspect, our semantic approach differs from the usual one where the involved class of modal algebras usually forms an equational class, i.e. a variety of algebras. Recall that a modal algebra in the usual sense is a Boolean algebra with an operator $f_\square$ satisfying the following sronger conditions for all elements $a,b$:
\begin{equation*}
\begin{split}
&f_\square(f_\wedge(a,b)) = f_\wedge(f_\square(a), f_\square(b))\text{ and }\\
&f_\square(f_\top)=f_\top.
\end{split}
\end{equation*}
It is known that the class of all modal algebras in this sense constitutes algebraic semantics for normal modal system $\mathit{K}$.

Given the modal language $Fm_\square$ and an $\mathit{S1SP}$-algebra $\mathcal{M}$, the notion of an assignment (valuation) $\gamma\colon V\rightarrow M$ is defined as before as a `homomorphism' from $Fm_\square$ to $\mathcal{M}$, in particular: $\gamma(\square\varphi)=f_\square(\gamma(\varphi))$. Also the notion of satisfaction is given in the same way: $(\mathcal{M},\gamma)\vDash\varphi\Leftrightarrow\gamma(\varphi)\in\mathit{TRUE}$. $\mathit{S1SP}$-algebras were introduced in \cite{lewjlc1} (not under this name) to provide a kind of algebraic semantics for Lewis-style modal logic $\mathit{S1SP}$:

\begin{theorem}[\cite{lewjlc1}]\label{600}
$\mathit{S1SP}$ is (strongly) sound and complete with respect to the class of all $\mathit{S1SP}$-algebras.
\end{theorem}

\begin{definition}\label{620}
A Boolean algebra expansion $\mathcal{M}$ is an $\mathit{S3}$-algebra if the following hold for all $a,b\in M$:\\
(1) $f_\square(a)\in\mathit{TRUE}\Leftrightarrow a=f_\top$\\
(2) $f_\square(a)\le a$\\
(S3) $f_\square(f_\rightarrow(a,b))\le f_\square(f_\rightarrow(f_\square(a), f_\square(b)))$
\end{definition}

\begin{lemma}\label{650}
Every $\mathit{S3}$-algebra is an $\mathit{S1SP}$-algebra, i.e. particularly condition (3) of Definition \ref{570} is satisfied. Moreover, in every $\mathit{S3}$-algebra, the modal operator $f_\square$ is a monotone function and it holds that
\begin{equation*}
f_\square(f_\wedge(a,b)) = f_\wedge(f_\square(a), f_\square(b)),
\end{equation*}
for all elements $a,b$.
\end{lemma}

\begin{proof}
Condition (S3) ensures that $f_\square$ is a monotone function: $a\le b$ iff $f_\rightarrow(a,b)=f_\top$ iff $f_\square(f_\rightarrow(a,b))\in\mathit{TRUE}$ $\overset{(S3)}{\Rightarrow}$ $f_\square(f_\rightarrow(f_\square(a), f_\square(b)))\in\mathit{TRUE}$ iff \\
$f_\rightarrow(f_\square(a), f_\square(b))=f_\top$ iff $f_\square(a)\le f_\square(b)$. Note that $f_\wedge(a,b)\le a$ and $f_\wedge(a,b)\le b$. Monotonicity implies 
\begin{equation*}
f_\square(f_\wedge(a,b))\le f_\wedge(f_\square(a),f_\square(b)).
\end{equation*}
On the other hand, $\varphi\rightarrow(\psi\rightarrow(\varphi\wedge\psi))$ is a propositional tautology and therefore denotes the top element, under any assignment. Thus, $f_\square(f_\rightarrow(a,f_\rightarrow(b,f_\wedge(a,b))))\in\mathit{TRUE}$, for any elements $a,b$. Condition (S3) along with `Modus Ponens' yields $f_\square(f_\rightarrow(f_\square(a),f_\square(f_\rightarrow(b,f_\wedge(a,b)))))\in\mathit{TRUE}$, i.e. $f_\square(a)\le f_\square(f_\rightarrow(b,f_\wedge(a,b)))$. Again by (S3), we get $f_\square(f_\rightarrow(b,f_\wedge(a,b)))\le  f_\square(f_\rightarrow(f_\square(b), f_\square(f_\wedge(a,b))))$. Thus, 
\begin{equation*}
\begin{split}
&f_\square(a)\le f_\square(f_\rightarrow(f_\square(b),f_\square(f_\wedge(a,b)))\le f_\rightarrow(f_\square(b),f_\square(f_\wedge(a,b))\text{ and hence}\\ 
&f_\rightarrow(f_\square(a),  f_\rightarrow(f_\square(b),f_\square(f_\wedge(a,b)))=f_\top.
\end{split}
\end{equation*}
The term on the left hand side of the last equation is an interpretation of the formula $x\rightarrow(y\rightarrow z)$ which is logically equivalent to $(x\wedge y)\rightarrow z$. Hence, $f_\rightarrow(f_\wedge(f_\square(a),f_\square(b)),f_\square(f_\wedge(a,b)))=f_\top$, i.e. 
\begin{equation*}
f_\wedge(f_\square(a),f_\square(b))\le f_\square(f_\wedge(a,b)).
\end{equation*}
Finally, $f_\wedge(f_\square(a),f_\square(b)) = f_\square(f_\wedge(a,b))$.\\ 
In order to see that every $\mathit{S3}$-algebra is an $\mathit{S1SP}$-algebra, it is enough to show that condition (3) of Definition \ref{570} follows from the conditions of Definition \ref{620}:\\ 
$((\varphi\rightarrow\psi)\wedge(\psi\rightarrow\chi))\rightarrow(\varphi\rightarrow\chi))$ is a propositional tautology and is therefore interpreted by the top element $f_\top$ of any model. By (1), 
\begin{equation*}
f_\square(f_\rightarrow(f_\wedge(f_\rightarrow(a,b),f_\rightarrow(b,c)), f_\rightarrow(a,c)))\in\mathit{TRUE}.
\end{equation*}
Applying (S3) and `Modus Ponens', we get 
\begin{equation*}
f_\square(f_\rightarrow(f_\square(f_\wedge(f_\rightarrow(a,b),f_\rightarrow(b,c))), f_\square (f_\rightarrow(a,c))))\in\mathit{TRUE}.
\end{equation*}
Since $f_\wedge(f_\square(a),f_\square(b)) = f_\square(f_\wedge(a,b))$, as shown above, we obtain the following: $f_\square(f_\rightarrow(f_\wedge(f_\square(f_\rightarrow(a,b)),f_\square(f_\rightarrow(b,c))), f_\square (f_\rightarrow(a,c))))\in\mathit{TRUE}$. Applying condition (1) yields
\begin{equation*}
f_\rightarrow(f_\wedge(f_\square(f_\rightarrow(a,b)),f_\square(f_\rightarrow(b,c))), f_\square (f_\rightarrow(a,c)))=f_\top,
\end{equation*}
i.e., $f_\wedge(f_\square(f_\rightarrow(a,b)),f_\square(f_\rightarrow(b,c)))\le f_\square(f_\rightarrow(a,c))$, which is precisely condition (3) of Definition \ref{570}.
\end{proof}

\begin{definition}\label{660}
We call a Boolean algebra expansion $\mathcal{M}$ a strong $\mathit{S4}$-algebra if the following conditions hold for all elements $a,b$:\\
(1) $f_\square(a)\in\mathit{TRUE}\Leftrightarrow a=f_\top$\\
(2) $f_\square(a)\le a$\\
(K)  $f_\square(f_\rightarrow(a,b))\le f_\rightarrow(f_\square(a), f_\square(b))$\\
(S4) $f_\square(a)\le f_\square (f_\square(a))$
\end{definition}

\begin{lemma}\label{670}
Every strong $\mathit{S4}$-algebra is an $\mathit{S3}$-algebra.
\end{lemma}

\begin{proof}
It is enough to show that condition (S3) holds in every strong $\mathit{S4}$-algebra. First, we observe that conditions (1) and (S4) imply that $f_\square(f_\top)=f_\top$. Then by (K), $f_\square(f_\rightarrow(f_\square(f_\rightarrow(a,b)),f_\rightarrow(f_\square(a), f_\square(b))))=f_\top$. Again by (K),\\ $f_\rightarrow(f_\square(f_\square(f_\rightarrow(a,b))), f_\square(f_\rightarrow(f_\square(a), f_\square(b))))=f_\top$. That is,\\
$f_\square(f_\square(f_\rightarrow(a,b))) \le f_\square(f_\rightarrow(f_\square(a), f_\square(b)))$. Applying condition (S4), we obtain condition (S3):
$f_\square(f_\rightarrow(a,b))) \le f_\square(f_\rightarrow(f_\square(a), f_\square(b)))$.
\end{proof}

If there is a notion of \textit{strong} $\mathit{S4}$-algebra, one may expect that there is a notion of $\mathit{S4}$-algebra, too. Indeed, $\mathit{S4}$-algebras have been studied in the literature under different labels such as \textit{topological Boolean algebras} or \textit{interior algebras}. An $\mathit{S4}$-algebra (alias interior algebra alias topological Boolean algebra) is usually defined as a Boolean algebra with an operator $f_\square$ (which can be viewed as an interior operator) such that  the following conditions (IA1)--(IA4) are satisfied for all elements $a,b$:\\

\noindent (IA1) $f_\square(a)\le a$\\
(IA2) $f_\square(f_\square(a))=f_\square(a)$\\
(IA3) $f_\square(f_\wedge(a,b)) = f_\wedge(f_\square(a), f_\square(b))$\\
(IA4) $f_\square(f_\top)=f_\top$.

\begin{theorem}\label{680}
Every strong $\mathit{S4}$-algebra is an $\mathit{S4}$-algebra.
\end{theorem}

\begin{proof}
Suppose $\mathcal{M}$ is a strong $\mathit{S4}$-algebra in the sense of Definition \ref{660}. Then (IA1) above holds trivially. (IA2) follows from (IA1) along with condition (S4). (IA3) is condition (S3) which holds by Lemma \ref{670}. By condition (1), $f_\square(f_\top)\in\mathit{TRUE}$. Then, by condition (S4), $f_\square(f_\square(f_\top))\in\mathit{TRUE}$. Again by (1), $f_\square(f_\top)=f_\top$, i.e. (IA4) is satisfied.
\end{proof}

The converse of Theorem \ref{680} is not true. As a contra-example we consider any interior algebra with more than two elements where the interior operator $f_\square$ is the identity: $a\mapsto f_\square(a)=a$. For every ultrafilter $U$, there exists an element $a\in U$ such that $a < f_\top$. Then condition (1) of Definition \ref{660} of a strong $\mathit{S4}$-algebra cannot be satisfied by all elements. An interior algebra gives rise to a strong $\mathit{S4}$-algebra if there is an ultrafilter $\mathit{TRUE}$ such that for any element $a$, $a< f_\top$ implies $f_\square(a)\notin\mathit{TRUE}$.\\
Thus, the class of strong $\mathit{S4}$-algebras is properly contained in the class of all $\mathit{S4}$-algebras. Nevertheless, for a completeness result concerning Lewis modal system $\mathit{S4}$, it is enough to consider only \textit{strong} $\mathit{S4}$-algebras.

\begin{definition}\label{700}
A Boolean algebra expansion $\mathcal{M}$ is called an $\mathit{S5}$-algebra if all elements $a$ satisfy the following:
\begin{equation*}
\begin{split}
f_\square(a)=
\begin{cases}
&f_\top,\text{ if } a=f_\top\\
&f_\bot,\text{ else}
\end{cases}
\end{split}
\end{equation*}
\end{definition}

Note that Definition $\ref{700}$ does not impose any condition on the designated ultrafilter $\mathit{TRUE}$ of the given Boolean algebra expansion. Actually, if we only consider the algebraic properties of an $\mathit{S5}$-algebra, then the designated ultrafilter can be disregarded. The resulting notion of an $\mathit{S5}$-algebra then is equivalent to the usual definitions of $\mathit{S5}$-algebras found in the literature. For example, an $\mathit{S5}$-algebra can be characterized as an interior algebra in which every open element is closed, i.e. where $f_\Diamond (f_\square(a))=f_\square(a)$ holds for every element $a$, with closure operator $f_\Diamond(a) := f_\neg (f_\square (f_\neg(a)))$. In fact, one easily verifies:

\begin{corollary}\label{710}
Let $\mathcal{M}$ be a Boolean algebra expansion. The following are equivalent:
\begin{itemize}
\item $\mathcal{M}$ is an $\mathit{S5}$-algebra.
\item $\mathcal{M}$ is an interior algebra satisfying for all $a\in M$: $f_\Diamond (f_\square(a))=f_\square(a)$.
\end{itemize}
\end{corollary}

In particular, every $\mathit{S5}$-algebra is an $\mathit{S4}$-algebra (i.e. an interior algebra). Given any $\mathit{S5}$-algebra, condition (1) of Definition \ref{660} is (trivially) satisfied, independently of the choice of the designated ultrafilter $\mathit{TRUE}$. Thus, every $\mathit{S5}$-algebra is also a strong $\mathit{S4}$-algebra.\\

Recall that the relation of satisfaction between $\mathit{S1SP}$-interpretations $(\mathcal{M},\gamma)$ and formulas $\varphi\in Fm_\square$ is given similarly as for $\mathit{SCI}$ models: $(\mathcal{M},\gamma)\vDash\varphi :\Leftrightarrow\gamma(\varphi)\in\mathit{TRUE}$. Also the concept of logical consequence is defined in the usual way. Extending the proof of Theorem \ref{600} in a straightforward way, we get

\begin{theorem}\label{720}
$\mathit{S3}$ ($\mathit{S4}$, $\mathit{S5}$) is strongly sound and complete w.r.t. the class of all $\mathit{S3}$-algebras ((strong) $\mathit{S4}$-algebras, $\mathit{S5}$-algebras), respectively.
\end{theorem}

\section{Dualities between $\mathit{SCI}$-theories and Lewis-style modal logics}

The goal of this section is to show that under certain assumptions, some Lewis-style modal logics are, in a precise sense, in duality with certain theories formalized in the language of $\mathit{SCI}$, more precisely, with certain axiomatic extensions of $\mathit{SCI}$. The crucial conditions for these dualities are the following:
\begin{enumerate}[(I)]
\item `The $\mathit{SCI}$ principles of propositional identity are valid. In particular, SP is valid.' 
\item `Propositional identity = strict equivalence', i.e., $(\varphi\equiv\psi)\equiv\square(\varphi\leftrightarrow\psi)$ holds.
\item `\textit{Necessity} = identity with proposition $\top$. In particular, there is exactly one necessary proposition: the proposition denoted by $\top$', i.e., $\square\varphi\equiv (\varphi\equiv\top)$ holds.
\item `All classical tautologies are necessary: If $\varphi$ is a classical tautology (i.e. an instance of a theorem of $\mathit{CPC}$), then $\square\varphi$ is valid.'
\end{enumerate}

From a semantic point of view, (III) and (IV) will ensure that the envolved $\mathit{SCI}$-models are Boolean algebras (cf. Theorem \ref{340} and the remark in the last paragraph of section 3.) \\

We remark here that a similar type of dualities between propositional logics with an identity connective and normal modal systems is established by T. Ishii \cite{ish}. His propositional calculus $\mathit{PCI}$ is also defined in the language of $\mathit{SCI}$ though the axioms (and rules) for the identity connective differ in some aspects from $\mathit{SCI}$. Ishii shows duality between $\mathit{PCI}$ and normal system $\mathit{K}$, as well as a series of further dualities between extensions of $\mathit{PCI}$ and corresponding normal modal systems.\footnote{Ishii does not use the term `duality'.}  \\

We now establish translations between the propositional languages of $\mathit{SCI}$ and of modal logic, i.e. between $Fm_\equiv$ and $Fm_\Box$.

\begin{definition}\label{1000}
The translation $\mathit{box}\colon Fm_\equiv\rightarrow Fm_\square$ is inductively defined as follows: $\mathit{box}(x):=x$, $\mathit{box}(\bot):=\bot$, $\mathit{box}(\top):=\top$, $\mathit{box}(\neg\varphi):=\neg \mathit{box}(\varphi)$, $\mathit{box}(\varphi *\psi):= (\mathit{box}(\varphi)*\mathit{box}(\psi))$, for $*\in\{\wedge,\vee,\rightarrow\}$, and 
\begin{equation*}
\mathit{box}(\varphi\equiv\psi):= (\square(\mathit{box}(\varphi)\rightarrow\mathit{box}(\psi))\wedge\square(\mathit{box}(\psi)\rightarrow\mathit{box}(\varphi))_.\footnote{We may abbreviate this by $\square(\mathit{box}(\varphi)\leftrightarrow\mathit{box}(\psi))$ since this formula is equivalent to the original one modulo $\mathit{SPS1}$, cf. Lemma \ref{580}.} 
\end{equation*}
On the other hand, the translation $\mathit{id}\colon Fm_\square\rightarrow Fm_\equiv$ is inductively defined as follows: $\mathit{id}(x):=x$, $\mathit{id}(\bot):=\bot$, $\mathit{id}(\top):=\top$, $\mathit{id}(\neg\varphi):=\neg \mathit{id}(\varphi)$, $\mathit{id}(\varphi *\psi):= (\mathit{id}(\varphi)*\mathit{id}(\psi))$, for $*\in\{\wedge,\vee,\rightarrow\}$, and
\begin{equation*}
\mathit{id}(\square\varphi):= (id(\varphi)\equiv\top).
\end{equation*}
For $\varPhi\subseteq Fm_\equiv$, we let $\mathit{box}(\varPhi):=\{\mathit{box}(\psi)\mid\psi\in\varPhi\}$; and for $\varPhi\subseteq Fm_\Box$, the set $\mathit{id}(\varPhi)$ is defined analogously.
\end{definition}

Induction on formulas ensures that $\mathit{box}(\varphi)\in Fm_\square$ for any $\varphi\in Fm_\equiv$; and $\mathit{id}(\varphi)\in Fm_\equiv$ for any $\varphi\in Fm_\square$. If the underlying logics are strong enough, then the translations $\mathit{box}$ and $\mathit{id}$ are inverse to each other in the sense of the next result.\\

Recall that in the language of modal logic $Fm_\square$, we use the following abbreviation: $(\varphi\equiv\psi) := (\square(\varphi\rightarrow\psi)\wedge \square (\psi\rightarrow\varphi))$, cf. \eqref{405} above. Since we are working with modal systems containing $\mathit{S1SP}$, we may define equivalently $(\varphi\equiv\psi) := \square(\varphi\leftrightarrow\psi)$, cf. Lemma \ref{580}. Also recall that in the language $\mathcal{L}_\equiv$ of $\mathit{SCI}$, we use the abbreviation $\square\varphi := (\varphi\equiv\top)$, cf. \eqref{330}.

\begin{theorem}\label{1020}
\begin{itemize}
\item Let $\mathcal{L}$ be a modal logic in the language $Fm_\square$ containing $\mathit{S1SP}$. 
Then for any $\varphi\in Fm_\square$: 
\begin{equation*}
\vdash_\mathcal{L}\varphi\equiv \mathit{box}(\mathit{id}(\varphi)).\footnote{In the following, we will write such an expression also as $\varphi\equiv_\mathcal{L} \mathit{box}(\mathit{id}(\varphi))$.}
\end{equation*}
\item Let $\mathcal{L}_\equiv$ be an axiomatic extension of $\mathit{SCI}$ in the language $Fm_\equiv$ containing theorems of the form $(\chi\equiv\psi)\equiv \square(\chi\leftrightarrow\psi)$.\footnote{That is, formulas of the form $(\chi\equiv\psi)\equiv ((\chi\leftrightarrow\psi)\equiv\top)$ are theorems.} Then for any $\varphi\in Fm_\equiv$:
\begin{equation*}
\vdash_{\mathcal{L}_\equiv}\varphi\equiv \mathit{id}(\mathit{box}(\varphi)).\footnote{We will write such an expression also as $\varphi\equiv_{\mathcal{L}_\equiv} \mathit{id}(\mathit{box}(\varphi))$.}
\end{equation*}
\end{itemize}
\end{theorem}

\begin{proof}
Under the assumptions of the first item, we show the assertion by induction on $\varphi\in Fm_\square$. If $\varphi$ is an atomic formula, we get $\mathit{box}(\mathit{id}(\varphi))=\varphi$. Then the assertion holds because $\square(\varphi\leftrightarrow\varphi)$ is a theorem of $\mathcal{L}$ (apply the rule of Axiom Necessitation (AN) to $\varphi\leftrightarrow\varphi$). Now suppose $\varphi = \square\psi$ for some $\psi\in Fm_\square$. 
\begin{equation*}
\begin{split}
\mathit{box}(\mathit{id}(\varphi)) &= \mathit{box}(\mathit{id}(\square\psi))\\ 
&=\mathit{box}(\mathit{id}(\psi)\equiv\top),\text{ by definition of }\mathit{id}\\
&=\square(\mathit{box}(\mathit{id}(\psi)\leftrightarrow\top)),\text{ by definition of }\mathit{box}\\
&\equiv_\mathcal{L}\square(\psi\leftrightarrow\top),\text{ by induction hypothesis and SP}\\
&=(\psi\equiv\top)\\
&\equiv_\mathcal{L}\square\psi,\text{ recall that }\square\psi\equiv(\psi\equiv\top)\text{ is a theorem of }\mathit{S1SP}\\
&=\varphi
\end{split}
\end{equation*}
Hence, $\varphi\equiv_\mathcal{L}\mathit{box}(\mathit{id}(\varphi))$, i.e. $\vdash_\mathcal{L}\varphi\equiv \mathit{box}(\mathit{id}(\varphi))$. The remaining cases of the induction step follow straightforwardly. Now, we assume the hypotheses of the second item and show its assertion by induction on $\varphi\in Fm_\equiv$. The induction base is clear; and in the induction step, only the case $\varphi=(\psi\equiv\chi)$ requires some attention:
\begin{equation*}
\begin{split}
\mathit{id}(\mathit{box}(\varphi))&=\mathit{id}(\mathit{box}(\psi\equiv\chi))\\
&=\mathit{id}(\square(\mathit{box}(\psi)\leftrightarrow\mathit{box}(\chi))),\text{ by definition of }\mathit{box}\\
&=(\mathit{id}(\mathit{box}(\psi))\leftrightarrow\mathit{id}(\mathit{box}(\chi)))\equiv\top,\text{ by definition of }\mathit{id}\\
&\equiv_{\mathcal{L}_\equiv} (\mathit{id}(\mathit{box}(\psi))\equiv\mathit{id}(\mathit{box}(\chi))), \text{ by assumptions on }\mathcal{L}_\equiv\\
&\equiv_{\mathcal{L}_\equiv} (\psi\equiv\chi),\text{ by induction hypothesis and SP}\\
&=\varphi
\end{split}
\end{equation*}
\end{proof}

If $\mathcal{L}$ is a modal logic and $\mathcal{L}_\equiv$ is an $\mathit{SCI}$-extension satisfying the hypotheses required in Theorem \ref{1020}, then we are able to establish a condition (actually, two equivalent conditions) under which both logics have, in a precise sense, the same expressive power, i.e. are dual to each other:

\begin{definition}\label{1030}
Let $\mathcal{L}$ be a modal logic in the language $Fm_\square$ containing $\mathit{S1SP}$. Let $\mathcal{L}_\equiv$ be an extension of $\mathit{SCI}$ in the language $Fm_\equiv$ containing theorems of the form $(\chi\equiv\psi)\equiv \square(\chi\leftrightarrow\psi)$. Furthermore, suppose one of the following two conditions is true:
\begin{enumerate}
\item For any $\varPhi\cup\{\varphi\}\subseteq Fm_\equiv$, $\varPhi\vdash_{\mathcal{L_\equiv}}\varphi \Longleftrightarrow \mathit{box}(\varPhi)\vdash_{\mathcal{L}}\mathit{box}(\varphi)$.
\item For any $\varPhi\cup\{\varphi\}\subseteq Fm_\square$, $\varPhi\vdash_{\mathcal{L}}\varphi \Longleftrightarrow \mathit{id}(\varPhi)\vdash_{\mathcal{L_\equiv}}\mathit{id}(\varphi)$.
\end{enumerate}
Then we say that $\mathcal{L}_\equiv$ and $\mathcal{L}$ are dual to each other, and we call $\mathcal{L}_\equiv$ the (dual) $\mathit{SCI}$-theory of modal logic $\mathcal{L}$; and we call $\mathcal{L}$ the (dual) modal theory of $\mathcal{L}_\equiv$. 
\end{definition}

Actually, it would be enough to consider only one of the conditions (i), (ii) in Definition \ref{1030}, as the next result shows.

\begin{lemma}\label{1035}
Let $\mathcal{L}$ be a modal logic and let $\mathcal{L}_\equiv$ be its dual $\mathit{SCI}$-theory according to Definition \ref{1030}. Then both conditions (i) and (ii) of Definition \ref{1030} are satisfied.
\end{lemma}

\begin{proof}
Let $\mathcal{L}_\equiv$ be the $\mathit{SCI}$-theory of modal system $\mathcal{L}$ and suppose that fact is witnessed by condition (i) of Definition \ref{1030}. We show that condition (ii) follows. Let $\varPhi\cup\{\varphi\}\subseteq Fm_\square$ and suppose $\varPhi\vdash_\mathcal{L}\varphi$. There are $\varphi_1,...,\varphi_n\in\varPhi$ such that $\vdash_\mathcal{L} (\varphi_1\wedge ...\wedge\varphi_n)\rightarrow\varphi$. By Theorem \ref{1020}, $\vdash_\mathcal{L}\mathit{box}(\mathit{id}((\varphi_1\wedge ...\wedge\varphi_n)\rightarrow\varphi))$. Then condition (i) yields $\vdash_{\mathcal{L}_\equiv} \mathit{id}((\varphi_1\wedge ...\wedge\varphi_n)\rightarrow\varphi)$. Taking into account the definition of $\mathit{id}$, that implies $\mathit{id}(\varPhi)\vdash_{\mathcal{L}_\equiv}\mathit{id}(\varphi)$. The implication from right-to-left of (ii) follows similarly. Analogously, one establishes condition (i) under the assumption that condition (ii) holds true.
\end{proof}

\begin{lemma}\label{1037}
Let $\mathcal{L}$ be a modal logic and let $\mathcal{L}_\equiv$ be its dual $\mathit{SCI}$-theory. Then the following hold:\\
(a) For any $\varphi\in Fm_\equiv$, $\vdash_\mathcal{L}\mathit{box}(\square\varphi)\equiv\square\mathit{box}(\varphi)$, i.e. $\mathit{box}(\square\varphi)\equiv_{\mathcal{L}}\square\mathit{box}(\varphi)$.\\
(b) For any $\varphi,\psi\in Fm_\square$, $\vdash_{\mathcal{L}_\equiv} \mathit{id}(\varphi\equiv\psi)\equiv (\mathit{id}(\varphi)\equiv\mathit{id}(\psi))$, which we also write as $\mathit{id}(\varphi\equiv\psi)\equiv_{\mathcal{L}_\equiv} (\mathit{id}(\varphi)\equiv\mathit{id}(\psi))$.
\end{lemma}

\begin{proof}
Under the given assumptions, we have: $\mathit{box}(\square\varphi)=\mathit{box}(\varphi\equiv\top)=\square(\mathit{box}(\varphi)\leftrightarrow\top)=(\mathit{box}(\varphi)\equiv\top)\equiv_\mathcal{L}\square\mathit{box}(\varphi)$. The last equation holds because $\square\psi\equiv (\psi\equiv\top)$ is a theorem of $\mathit{S1SP}$ and thus of $\mathcal{L}$, for any $\psi\in Fm_\square$.\\
On the other hand: $\mathit{id}(\varphi\equiv\psi)=\mathit{id}(\square(\varphi\leftrightarrow\psi))=(\mathit{id}(\varphi)\leftrightarrow\mathit{id}(\psi))\equiv\top)\equiv_{\mathcal{L}_\equiv} (\mathit{id}(\varphi)\equiv\mathit{id}(\psi))$. The last equation holds because formulas of the form $(\chi\equiv\xi)\equiv ((\chi\leftrightarrow\xi)\equiv\top)$ are theorems of $\mathcal{L}_\equiv$. 
\end{proof}

As expected, particular examples of Definition \ref{1030} are the $\mathit{SCI}$-theories of modal systems $\mathit{S1SP}$, $\mathit{S3}$, $\mathit{S4}$ and $\mathit{S5}$ which we are going to define in the following as deductive systems in the language of $\mathit{SCI}$. Recall that we have $\square\varphi := (\varphi\equiv\top)$.

\begin{definition}\label{1040}
We consider the language $Fm_\equiv$ of $\mathit{SCI}$ and define deductive systems on the base of the following axiom schemes (CPC) + (1)--(5):\\
(CPC) any formula $\varphi$ having the form of a classical tautology, i.e. $\varphi$ is the substitution instance of a theorem of $\mathit{CPC}$\\
(1) $(\chi\equiv\psi)\leftrightarrow \square(\chi\leftrightarrow\psi)$\\
(2) $\square\varphi\rightarrow\varphi$\\
(3')$(\square(\varphi\rightarrow\psi)\wedge \square(\psi\rightarrow\chi))\rightarrow\square(\varphi\rightarrow\chi)$\\
(3) $\square(\varphi\rightarrow\psi)\rightarrow\square (\square\varphi\rightarrow\square\psi)$\\
(4) $\square\varphi\rightarrow\square\square\varphi$\\
(5) $\neg\square\varphi\rightarrow\square\neg\square\varphi$.\\
Then logic $\mathit{S1SP}_\equiv$ is axiomatized by the axiom schemes (CPC), (1), (2), (3') together with the scheme of theorems SP $(\varphi\equiv\psi)\rightarrow (\chi[x:=\varphi]\equiv\chi[x:=\psi])$. That is, $\mathit{S1SP}_\equiv$ is given by the following deductive system. For $\varPhi\cup\{\varphi\}\subseteq Fm_\equiv$, we write $\varPhi\vdash_{\mathit{S1SP}_\equiv}\varphi$ if there is a derivation, i.e. a sequence $\varphi_1,...,\varphi_n=\varphi$, such that for every $\varphi_i$, $1\le i\le n$: $\varphi_i\in\varPhi$ or $\varphi_i$ is an instance of (CPC), (1)--(3') or SP or $\varphi_i$ is obtained by rule MP or $\varphi_i$ is obtained by rule AN (i.e. there is some $1\le j<i$ such that $\varphi_j$ is an axiom, i.e. an instance of (CPC) + (1)--(3') and $\varphi_i=\square\varphi_j$).\\
The deductive system $\mathit{S3}_\equiv$ is defined analogously but with axiom schemes (CPC), (1), (2), (3) (and without theorem scheme SP). Similarly, logic $\mathit{S4}_\equiv$ is given by the axioms (CPC) and (1)--(4). If additionally we consider axiom scheme (5), then we obtain system $\mathit{S5}_\equiv$.\footnote{Of course, rule AN only applies to the given axioms of the respective underlying system.}
\end{definition}

\begin{lemma}\label{1042}
$(\chi\equiv\psi)\equiv \square(\chi\leftrightarrow\psi)$ is a theorem of $\mathit{S1SP}_\equiv$.
\end{lemma}

\begin{proof}
Applying rule AN to (1) results in $\square((\chi\equiv\psi)\leftrightarrow \square(\chi\leftrightarrow\psi))$. Formula $((\chi\equiv\psi)\equiv \square(\chi\leftrightarrow\psi))\leftrightarrow \square((\chi\equiv\psi)\leftrightarrow \square(\chi\leftrightarrow\psi))$ is an instance of (1). Modus Ponens yields $(\chi\equiv\psi)\equiv \square(\chi\leftrightarrow\psi)$.
\end{proof}

\begin{theorem}\label{1050}
$\mathit{SCI}\subseteq\mathit{SCI^+}\subseteq\mathit{S1SP}_\equiv\subseteq\mathit{S3}_\equiv\subseteq\mathit{S4}_\equiv\subseteq\mathit{S5}_\equiv$.
\end{theorem}

\begin{proof}
The first inclusion is trivial by the definitions (cf. Definition \ref{350}).\\
\textbf{Claim 1}: $\mathit{SCI}^+\subseteq\mathit{S1SP}_\equiv$.\\
It is enough to show $\mathit{SCI}\subseteq\mathit{S1SP}_\equiv$. Recall that SP is euivalent to the identity axioms (id3)--(id7) (modulo the rest of $\mathit{SCI}$). So we only need to show that (id1) $\varphi\equiv\varphi$ and (id2) $(\varphi\equiv\psi)\rightarrow(\varphi\rightarrow\psi)$ are theorems of $\mathit{S1SP}_\equiv$. (id1) derives considering axiom $\varphi\leftrightarrow\varphi$, rule AN and scheme (1). (id2) derives from (1)+(2). Thus Claim 1 is true.\\
\textbf{Claim 2}: $\square(\varphi\wedge\psi)\rightarrow (\square\varphi\wedge\square\psi)$ is a theorem of $\mathit{S3}_\equiv$.\\
Apply AN to the tautologies $(\varphi\wedge\psi)\rightarrow\varphi$ and $(\varphi\wedge\psi)\rightarrow\psi$ and consider axiom schemes (3) and (2). Using propositional calculus, Claim 3 follows.\\
\textbf{Claim 3}: $\mathit{S1SP}_\equiv\subseteq\mathit{S3}_\equiv$.\\ 
It is enough to show that scheme (3) is stronger than (3'), and that scheme SP is derivable in $\mathit{S3}_\equiv$. Of course, $(\varphi\rightarrow\psi)\rightarrow ((\psi\rightarrow\chi)\rightarrow (\varphi\rightarrow\chi))$ is a propositional tautology and thus an axiom. Applying rule AN, (3), (2) and modus ponens then yields $\square(\varphi\rightarrow\psi)\rightarrow (\square(\psi\rightarrow\chi)\rightarrow\square(\varphi\rightarrow\chi))$. Modulo $\mathit{CPC}$, this is equivalent to (3'). (Note that we argued as in original modal logic.) Thus, (3) is stronger than (3') (modulo the rest). Finally, in order to show that principle SP is derivable, we derive the identity axioms (id3)--(id7) of $\mathit{SCI}$ which are equivalent to SP modulo the rest. Consider the tautology $(\varphi\leftrightarrow\psi)\rightarrow (\neg\varphi\leftrightarrow\neg\psi)$ and apply AN, (3), (2) and MP. We derive $\square(\varphi\leftrightarrow\psi)\rightarrow \square(\neg\varphi\leftrightarrow\neg\psi)$. By scheme (1) and transitivity of implication, we get $(\varphi\equiv\psi)\rightarrow (\neg\varphi\equiv\neg\psi)$, i.e. (id3). Now we consider the tautology $(\varphi\leftrightarrow\psi)\rightarrow ((\varphi'\leftrightarrow\psi')\rightarrow ((\varphi \vee \varphi')\leftrightarrow(\psi \vee \psi')))$. By AN and axioms, $\square(\varphi\leftrightarrow\psi)\rightarrow (\square(\varphi'\leftrightarrow\psi')\rightarrow \square((\varphi \vee \varphi')\leftrightarrow(\psi \vee \psi')))$. In this formula, we may replace formulas of the form $\square(\chi_1\leftrightarrow\chi_2)$ by $\chi_1\equiv\chi_2$, according to (1). This results in $(\varphi\equiv\psi)\rightarrow ((\varphi'\equiv\psi')\rightarrow ((\varphi \vee \varphi')\equiv(\psi \vee \psi')))$ which is equivalent to $((\varphi\equiv\psi)\wedge (\varphi'\equiv\psi'))\rightarrow ((\varphi \vee \varphi')\equiv(\psi \vee \psi'))$, i.e. (id4). Similarly, we derive (id5) and (id6). Towards (id7), we consider the propositional tautology\\ 
$(\varphi\leftrightarrow\psi)\rightarrow ((\varphi'\leftrightarrow\psi')\rightarrow ((\varphi \leftrightarrow \varphi')\leftrightarrow(\psi \leftrightarrow \psi')))$ and derive\\
(*) $\square(\varphi\leftrightarrow\psi)\rightarrow (\square(\varphi'\leftrightarrow\psi')\rightarrow \square((\varphi \leftrightarrow \varphi')\leftrightarrow(\psi \leftrightarrow \psi')))$ in a similar way as before. Using Claim 2 and axiom scheme (3), we get\\
$\square((\varphi \leftrightarrow \varphi')\leftrightarrow(\psi \leftrightarrow \psi'))\rightarrow \square(\square(\varphi \leftrightarrow \varphi')\leftrightarrow \square (\psi \leftrightarrow \psi'))$. Considering (*) and transitivity of implication, we derive\\
$\square(\varphi\leftrightarrow\psi)\rightarrow (\square(\varphi'\leftrightarrow\psi')\rightarrow  \square(\square(\varphi \leftrightarrow \varphi')\leftrightarrow \square (\psi \leftrightarrow \psi')))$. Now, in the same way as before, we apply (1) and corresponding replacements to derive\\
(**) $(\varphi\equiv\psi)\rightarrow ((\varphi'\equiv\psi')\rightarrow  (\square(\varphi \leftrightarrow \varphi')\equiv \square(\psi \leftrightarrow \psi')))$. Note that the proof of Lemma \ref{1042} also works in $\mathit{S3}_\equiv$. By schemes (3') and (1), the connective $\equiv$ is transitive in $\mathit{S3}_\equiv$. Putting these observations together and considering the equations `$(\varphi \equiv \varphi')\equiv \square(\varphi \leftrightarrow \varphi')\equiv\square(\psi \leftrightarrow \psi')\equiv (\psi \equiv \psi')$', we are able to derive\\
$(\square(\varphi \leftrightarrow \varphi')\equiv \square(\psi \leftrightarrow \psi'))\rightarrow ((\varphi \equiv \varphi')\equiv (\psi \equiv \psi'))$. This together with (**) and transitivity of implication yields\\
$(\varphi\equiv\psi)\rightarrow ((\varphi'\equiv\psi')\rightarrow  ((\varphi \equiv \varphi')\equiv (\psi \equiv \psi')))$ which is equivalent to (id7). Thus, Claim 3 is true. Finally, the inclusions $\mathit{S3}_\equiv\subseteq\mathit{S4}_\equiv\subseteq\mathit{S5}_\equiv$ are clear by Definition \ref{1040}.
\end{proof}

We are now able to establish the intended dualities between some of our $\mathit{SCI}$-theories and corresponding modal systems.

\begin{theorem}\label{1060}
The logics $\mathit{S1SP}_\equiv$, $\mathit{S3}_\equiv$, $\mathit{S4}_\equiv$ and $\mathit{S5}_\equiv$ introduced in Definition \ref{1040} are the dual $\mathit{SCI}$-theories of the modal logics $\mathit{S1SP}$, $\mathit{S3}$, $\mathit{S4}$ and $\mathit{S5}$, respectively.
\end{theorem}

\begin{proof}
We prove the duality between $\mathit{S3}$ and $\mathit{S3}_\equiv$. The remaining dualities follow in the same way. First, let us check that the logics $\mathcal{L}:=\mathit{S3}$ and $\mathcal{L}_\equiv :=\mathit{S3}_\equiv$ satisfy the conditions of Definition \ref{1030}. On the one hand, we know that $\mathit{S3}$ is the weakest Lewis modal system containing principle SP (c.f. \cite{lewjlc1, lewsl}) and thus contains $\mathit{S1SP}$. On the other hand, by Theorem \ref{1050} and Lemma \ref{1042}, we know that $\mathcal{L}_\equiv=\mathit{S3}_\equiv$ contains $\mathit{SCI}$ and theorems $(\chi\equiv\psi)\equiv \square(\chi\leftrightarrow\psi)$. It remains to check one of the equivalent conditions (i) or (ii) of Definition \ref{1030}. We show that (i) holds. So let $\varPhi\cup\{\varphi\}\subseteq Fm_\equiv$ and suppose $\varPhi\vdash_{\mathit{S3}_\equiv}\varphi$. We show $\mathit{box}(\varPhi)\vdash_{\mathit{S3}}\mathit{box}(\varphi)$ by induction on the length $n\ge 1$ of derivations of $\varphi$ from $\varPhi$ in $\mathit{S3}_\equiv$. If $n=1$, then we distinguish the following cases (a)--(d).\\
(a) $\varphi\in\varPhi$. Then trivially $\mathit{box}(\varphi)\in\mathit{box}(\varPhi)$ and thus $\mathit{box}(\varPhi)\vdash_{\mathit{S3}}\mathit{box}(\varphi)$.\\
(b) $\varphi$ has the form of a classical tautology. Since translation $\mathit{box}$ preserves logical connectives, it follows that $\mathit{box}(\varphi)$ is of the same form, i.e., has the form of a classical tautology, too, and as such is an axiom of $\mathit{S3}$.\\ 
(c) $\varphi$ is an instance of scheme (1), say $\varphi = (\chi\equiv\psi)\leftrightarrow ((\chi\leftrightarrow\psi)\equiv\top)$. By definition of $\mathit{box}$:\\ 
$\mathit{box}(\varphi)=\square(\mathit{box}(\chi)\leftrightarrow\mathit{box}(\psi))\leftrightarrow \square((\mathit{box}(\chi)\leftrightarrow\mathit{box}(\psi))\leftrightarrow\top)$.\\
Considering the definition of the identity connective $(\varphi_1\equiv\varphi_2) :=\square (\varphi_1\leftrightarrow\varphi_2)$ in $\mathit{S3}$, this yields $\mathit{box}(\varphi)=(\mathit{box}(\chi)\equiv\mathit{box}(\psi))\leftrightarrow ((\mathit{box}(\chi)\leftrightarrow\mathit{box}(\psi))\equiv\top)$.\\
By Lemma \ref{430}, $\square(\mathit{box}(\chi)\leftrightarrow\mathit{box}(\psi)) \equiv ((\mathit{box}(\chi)\leftrightarrow\mathit{box}(\psi))\equiv\top)$ is a theorem of $\mathit{S3}$. Applying SP, we get
\begin{equation*}
\begin{split}
\mathit{box}(\varphi)&\equiv_{\mathit{S3}}((\mathit{box}(\chi)\equiv\mathit{box}(\psi))\leftrightarrow\square(\mathit{box}(\chi)\leftrightarrow\mathit{box}(\psi)))\\
&=(\mathit{box}(\chi)\equiv\mathit{box}(\psi))\leftrightarrow(\mathit{box}(\chi)\equiv\mathit{box}(\psi)).
\end{split}
\end{equation*}
Of course, any such trivial biconditional is a theorem of $\mathit{S3}$ and so is $\mathit{box}(\varphi)$.\\
(d) $\varphi$ is an instance of scheme (2), say $\varphi = (\square\psi\rightarrow\psi)$. By Lemma \ref{1037}(a), $\mathit{box}(\varphi)\equiv_{\mathit{S3}}\square\mathit{box}(\psi)\rightarrow\mathit{box}(\psi)$. The latter is an axiom of $\mathit{S3}$. \\
(e) $\varphi$ is an instance of scheme (3), say $\varphi = \square(\psi\rightarrow\chi)\rightarrow\square (\square\psi\rightarrow\square\chi)$. As in (d), we apply Lemma \ref{1037}(a) and get\\ 
$\mathit{box}(\varphi)\equiv_{\mathit{S3}}\square(\mathit{box}(\psi)\rightarrow\mathit{box}(\chi))\rightarrow\square (\square\mathit{box}(\psi)\rightarrow\square\mathit{box}(\chi))$. The latter is an axiom of $\mathit{S3}$.\footnote{Note that the same argument is applicable if we consider the axioms (3'), (4), (5). If $\varphi$ is such an axiom, then $\mathit{box}(\varphi)$ is the corresponding axiom of modal system $\mathit{S1SP}$, $\mathit{S4}$, $\mathit{S5}$, respectively.}  \\
Examining the cases (b)--(e) above, we conclude in particular the following\\
\textbf{Fact}: For any axiom $\chi$ of $\mathit{S3}_\equiv$, we have $\mathit{box}(\chi)\equiv_\mathit{S3}\chi'$, where $\chi'$ is an axiom of modal system $\mathit{S3}$.\\
Now, suppose $\varphi$ is derived in $n+1$ steps and the assertion is true for all derivations of length $\le n$. We may assume that $\varphi$ is obtained by an application of the rules MP or AN. In the former case, there are $\psi$ and $\psi\rightarrow\varphi$ derived in $\le n$ steps, and the induction hypothesis yields $\mathit{box}(\varPhi)\vdash_{\mathit{S3}}\mathit{box}(\varphi)$. In the latter case, $\varphi=\square\chi$ for some axiom $\chi$ of $\mathit{S3}_\equiv$ that occurs in the given derivation. By Lemma \ref{1037}(a) and the \textbf{Fact} above, $\mathit{box}(\varphi)\equiv_\mathit{S3}\square\mathit{box}(\chi)$ and $\mathit{box}(\chi)\equiv_\mathit{S3}\chi'$, where $\chi'$ is an axiom of modal system $\mathit{S3}$. Since SP holds in $\mathit{S3}$, we may replace $\mathit{box}(\chi)$ by $\chi'$ in every context. Applying SP in $\mathit{S3}$, we get $\mathit{box}(\varphi)\equiv_\mathit{S3}\square\chi'$. Since $\chi'$ is an axiom of $\mathit{S3}$, formula $\square\chi'$ is a theorem of $\mathit{S3}$ by the rule of Axiom Necessitation. Hence, $\mathit{box}(\varphi)$ is a theorem of $\mathit{S3}$. We have finished the induction and thus the proof of the Theorem.
\end{proof}

We have established dualities between some particular $\mathit{SCI}$-theories and corresponding Lewis-style modal logics by means of the respective deductive systems (cf. Definition \ref{1030}). How can these dualities be described semantically? One easily recognizes that a given $\mathit{S1SP}$-algebra can be transformed into an $\mathit{SCI}$-model defining $f_\equiv(a,b):=f_\square(f_\leftrightarrow(a,b))$, where $f_\leftrightarrow(a,b)$ is defined in the obvious way. This corresponds to the theorem $(\varphi\equiv\psi)\equiv \square(\varphi\leftrightarrow\psi)$ of $\mathit{S1SP}$. The resulting $\mathit{SCI}$-model then will be a model of $\mathit{S1SP}_\equiv$. The other way round, any given $\mathit{SCI}$-model which is a model of $\mathit{S1SP}_\equiv$ can be transformed into an $\mathit{S1SP}$-algebra defining $f_\square(a):=f_\equiv(a,f_\top)$. This corresponds to the theorem $\square\varphi\equiv (\varphi\equiv\top)$ of modal system $\mathit{S1SP}$. We conclude that the $\mathit{SCI}$-theory $\mathit{S1SP}_\equiv$ is sound and complete w.r.t. the class of exactly those $\mathit{SCI}$-models which can be obtained from $\mathit{S1SP}$-algebras by the above presented transformation. So from a semantic point of view, the duality between $\mathit{SCI}$-theory $\mathit{S1SP}_\equiv$ and modal system $\mathit{S1SP}$ is given by those respective classes of models (and the transformations in both directions). Analogously, we can describe the remaining dualities semantically. Detailed proofs derive straightforwardly from the above results.

Our view on intensionality as a measure for the discernibility of propositions (`the more propositions can be distinguished in models of the underlying logic the higher degree of intensionality') is presented here in a rather informal and intuitive way. An interesting task for future work could be a precise formalization of that concept -- in classical as well as in non-classical settings. The dualities established in this paper generalize and extend earlier results (e.g. \cite{blosus, lewjlc1}) or are in analogy with similar results that hold in propositional logics distinct from $\mathit{SCI}$ (cf. \cite{ish}). The question arises which further (hyper-) intensional logics can be represented in a framework based on $\mathit{SCI}$ or based on a logic with different axioms for propositional identity. Can all (hyper-) intensional logics be captured by an appropriate axiomatization of propositional identity? These and similar questions remain to be further investigated.

\end{document}